%% file: main.tex
\documentclass[acmsmall,nonacm]{acmart}

\newif\ifpublish
\publishfalse  

\author{Rithwik Kerur}
\email{rkerur@ucsb.edu}
\affiliation{%
  \institution{University of California, Santa Barbara}
  \city{Santa Barbara}
  \state{CA}
  \country{USA}
}

\author{Divyakant Agrawal}
\email{divyagrawal@ucsb.edu}
\affiliation{%
  \institution{University of California, Santa Barbara}
  \city{Santa Barbara}
  \state{CA}
  \country{USA}
}

\author{Dahlia Malkhi}
\email{dahliamalkhi@ucsb.edu}
\affiliation{%
  \institution{University of California, Santa Barbara}
  \city{Santa Barbara}
  \state{CA}
  \country{USA}
}

\author{Michael K. Reiter}
\email{michael.reiter@duke.edu}
\affiliation{%
  \institution{Duke University}
  \city{Durham}
  \state{NC}
  \country{USA}
}

\author{Amit Wieder}
\email{wiederamit@gmail.com}
\affiliation{%
  \institution{University of Colorado, Boulder}
  \city{Boulder}
  \state{CO}
  \country{USA}
}
\input{macros}

\title{\sys: Elastic Asynchronous Information Dispersal with Post-Dissemination Pruning}

\begin{document}

\input{abstract}

\maketitle

\clearpage
\setcounter{page}{1}
\input{intoduction-Rithwik}
\input{background}
\input{problem}

\input{protocol}
\input{useCase}

\input{implementation}

\input{evaluation}

\input{RelatedWork}
\input{conclusion}

\newpage
\bibliographystyle{plain}
\bibliography{references}

\newpage
\input{Appendix/appendix}
\end{document}

%% file: macros.tex
\usepackage{graphicx} % Required for inserting images

\usepackage{float}
\usepackage{amssymb}   
\usepackage{xcolor}
\usepackage{xspace} 
\usepackage{comment}
\usepackage{amsmath}
\usepackage[normalem]{ulem}
\usepackage{subcaption}
\usepackage{wrapfig}
\usepackage{algorithm}
\usepackage[inline]{enumitem}
\usepackage{algpseudocode}

\newcommand{\sys}{eAID\xspace}
\newcommand{\sysConsensus}{eAID-KV\xspace}
\newcommand{\HRaft}{HRaft\xspace}
\newcommand{\CRaft}{CRaft\xspace}
\newcommand{\FlexRaft}{FlexRaft\xspace}
\newcommand{\FlexRaftIDA}{Proactive Encoding\xspace}
\newcommand{\KV}{key-value store\xspace}
\newcommand{\KVs}{key-value stores\xspace}
\newcommand{\LogEntry}{log entry\xspace}
\newcommand{\LogEntries}{log entries\xspace}
\newcommand{\RQ}{responsive quorum\xspace}
\newcommand{\RQcapital}{Responsive-quorum\xspace}

\newcommand{\omitit}[1]{}

\ifpublish
    \newcommand{\DM}[1]{}
    \newcommand{\MKR}[1]{}
    \newcommand{\TODO}[1]{}
    \newcommand{\Rithwik}[1]{}
\else
    \newcommand{\DM}[1]{{\color{blue} DM: #1}}
    \newcommand{\MKR}[1]{{\color{purple} MKR: #1}}
    \newcommand{\TODO}[1]{{\color{red} TODO: #1}}
    \newcommand{\Rithwik}[1]{{\color{red} Rithwik: #1}}
\fi

%% file: abstract.tex
\begin{abstract}
Spreading and storing erasure-coded data effectively in distributed systems is challenging in practical settings. Consider a system of
$N$ nodes designed to tolerate up to $F<N/2$ failures. The dissemination of erasure-coded information is typically designed to complete only after receiving messages from $(N-F)$ nodes, thereby preparing for the worst-case, but rare, scenario of $F$ failures. In steady state, the remaining $F$ nodes may in fact be healthy, but their resources are not counted. This leads to over-provisioning of storage for encoded data.

This paper introduces \textbf{\sys}, a novel \textit{elastic information dispersal} algorithm that addresses this conundrum through a two-stage approach.

First, the core protocol estimates the actual number $f$ of faulty nodes, rather than assuming the worst-case bound $F$. Dissemination completes quickly when messages are received from $(N-f)$ nodes, and more gradually when fewer nodes respond. This flexibility enables the system to tune storage efficiency against latency.
Second, after initial dissemination completes, \sys continues monitoring for additional responses. As responses arrive from up to $N$ nodes, the system prunes the information stored at responding nodes accordingly.

A key technique enabling this seamless elasticity is an agile encoding scheme that varies the number of disseminated fragments while keeping both fragment size and the recovery threshold $(F+1)$ fixed. Not only does this enable varying the number of disseminated fragments on the fly, it also allows nodes to discard encoded fragments autonomously. Crucially, this is achieved without maintaining complex metadata, without requiring nodes to reconstruct or re-encode information, and without global coordination for storage decisions.

We demonstrate the practicality of \sys by integrating it with a replicated \KV and evaluating it in network environments with unpredictable latencies. The results show that \sys improves overall performance while significantly reducing long-term storage consumption.
\end{abstract}

%% file: intoduction-Rithwik.tex
\section{Introduction}
The challenge of reliably storing and disseminating information in distributed systems is often a trade-off between fault tolerance and resource efficiency. While replication provides high availability, it incurs significant storage and communication overheads. Information dispersal algorithms, pioneered by Rabin's IDA~\cite{Rabin}, offer a compelling alternative by splitting data into fragments using erasure codes (EC), allowing for reconstruction from a subset of fragments. In theory, IDA provides optimal redundancy. In practice, however, deploying it within critical infrastructure presents significant operational challenges. 

The primary difficulty lies not in the coding theory, as robust EC schemes are well-established \cite{LubyTransform, FountainCodes, Reed-Solomon}, but in handling unpredictable network latencies. 
In \textit{asynchronous} settings, IDA methods must complete disseminating encoded fragments to $N$ nodes in face of up to $F < N/2$ unresponsive ones. 
The problem is that the system cannot distinguish between a slow node and a crashed node. 

To address this challenge, prior IDA-based systems have proposed several adaptive dissemination strategies: 
falling back to full replication upon timer expiration, as in RS-Paxos~\cite{RS-Paxos} and \CRaft~\cite{CRaft}; 
re-sharing ``endangered'' fragments to responsive nodes, as in OceanStore~\cite{OceanStore} and \HRaft~\cite{HRaft}; 
and dynamically varying the encoding threshold, as in \FlexRaft~\cite{FlexRaft}.

While these approaches improve dissemination under unstable network conditions, they suffer from two fundamental limitations. First, preserving optimal long-term storage cost is operationally burdensome. Temporary transmission delays during dissemination may cause the system to increase redundancy, leaving nodes with excess encoded data even after network conditions stabilize. Restoring storage efficiency then requires reconstructing and re-disseminating information, and typically relies on centralized coordination to determine storage decisions across nodes.

Second, these approaches complicate recovery. To support adaptive coding, they maintain per-entry metadata describing the specific coding configuration used for each entry, and recovery must therefore apply entry-specific reconstruction procedures and parameters.

\medskip

In this paper, we introduce \textbf{\sys},
an elastic way to AID in Asynchronous Information Dispersal, addressing these shortcomings.
The core idea in \sys is to employ an erasure code scheme that varies the \textbf{number} of disseminated fragments while keeping both fragment size and the recovery threshold $(F+1)$ fixed. 
This technique fundamentally rethinks how erasure coding is managed in dynamic networks and enables two-prong, seamless elasticity:

\begin{enumerate}

\item 
The dissemination protocol estimates the actual number $f$ of faulty nodes, rather than assuming the worst-case bound $F$. Dissemination in \sys tunes the number of disseminated fragments on the fly based on the estimate $f$. It completes quickly when messages are received from $(N-f)$ nodes, and more gradually when fewer nodes respond. This flexibility enables the system to tune storage efficiency against latency.

\item 
After initial dissemination completes, \sys continues monitoring for additional responses. As responses arrive from up to $N$ nodes, the system prunes the information stored at responding nodes accordingly.
Note that the \sys coding strategy allows nodes to discard encoded fragments autonomously. Crucially, this is achieved without maintaining complex metadata, without requiring nodes to reconstruct or re-encode information, and without global coordination for storage decisions.
Thus, unlike prior works that focus exclusively on the dissemination phase, \sys is the first solution to incorporate \textit{post-dissemination pruning}.

\end{enumerate}

%Concretely, this elasticity comes from a fixed $(F+1,\,(F+1)\times(N-1))$ code. 
Concretely, this elasticity comes from a fixed $(k = (F+1),\,m = (F+1)\times(N-1))$ code\footnote{$k$ denotes the number of fragments required for reconstruction; $m$ the count of extra fragments.}. 
The leader draws from a large pool of distinct fragments and assigns more of them per node as its
failure estimate grows, so that $F+1$ fragments always survive among the responsive
nodes. When responses are slow, \sys can avoid a
retransmission round by overestimating failures and sending a few extra
fragments, then pruning the surplus as soon as acknowledgments arrive. This trades a
little dissemination bandwidth, which is typically abundant, for lower latency, which is
often critical, particularly for mission-critical data-center and cloud
systems~\cite{TailAtScale}. We detail these strategies in Section~\ref{sec:Solution}.

Because the recovery threshold is fixed, \sys can reconstruct every disseminated information item from any
$F+1$ fragments using the same procedure, regardless of how that item was disseminated or later pruned. Recovery is therefore uniform and needs none of the
per-entry coding metadata that other schemes must need to track.

We integrate \sys into a \KV and evaluate the resulting system, \sysConsensus, under
network settings with highly variable point-to-point latencies characteristic of shared cloud environments~\cite{Pingmesh2015}. \sys matches or exceeds the low-latency dissemination of prior protocols with no metadata overhead and no re-encoding, while \sysConsensus uniquely reclaims long-term storage as network health improves, incurring only a minimal and temporary storage increase during disruptions. By optimizing for the
critical path first and storage second, the evaluation substantiates our motivation for creating a
robust, and storage-efficient dissemination regime for a replicated \KV.

%% file: background.tex
\section{Background}
\label{sec:Background}

\subsection{Information Dispersal}

The information dispersal problem, originally formulated in Rabin's seminal IDA method~\cite{Rabin}, considers the following challenge: a leader has a message $M$ of size $B$ that it intends to disseminate among a set of $N$ nodes, while providing both data availability and storage efficiency. Specifically, the dispersal must guarantee that despite up to $F < N/2$ nodes crashing, the message $M$ can be successfully reconstructed from the information stored on nodes that haven't failed. 

Rather than sending each of the $N$ nodes a full replica of $M$, a key technique for storage efficiency is employing erasure codes. 
An erasure code is parameterized using two values, $(k, m)$: it encodes a message $M$ into $k+m$ code pieces referred to as \textit{fragments}, such that any $k$ fragments suffice to reconstruct $M$.
In the encoded information dispersal approach, the leader encodes $M$ via a $(k, m)$ erasure code and disseminates fragment(s) to nodes.
With a Reed-Solomon code~\cite{Reed-Solomon}, 
each fragment size is $B/k$, the total storage is $\frac{B}{k}\times (k+m)$, and the \textit{storage blowup} over $B$ is $\frac{(k+m)}{k}$. 
Other practical codes, such as Luby-Transform \cite{LubyTransform} or Fountain codes \cite{Digital-Fountains} have similar blowup characteristics. 

In Rabin's IDA, the parameters $(k,m)$ of the erasure code directly mirror the system parameters $N=2F+1$, with $k=F+1, m=F$. With this parameterization, there are $k+m = 2F+1 = N$ fragments, and a leader assigns exactly one fragment to each node. This assignment maintains data availability against up to $F$ node failures with a factor of $(2F+1)/(F+1) \approx 2$ storage blowup.

\begin{figure}[tb]
    \centering
    \includegraphics[width=0.7\linewidth]{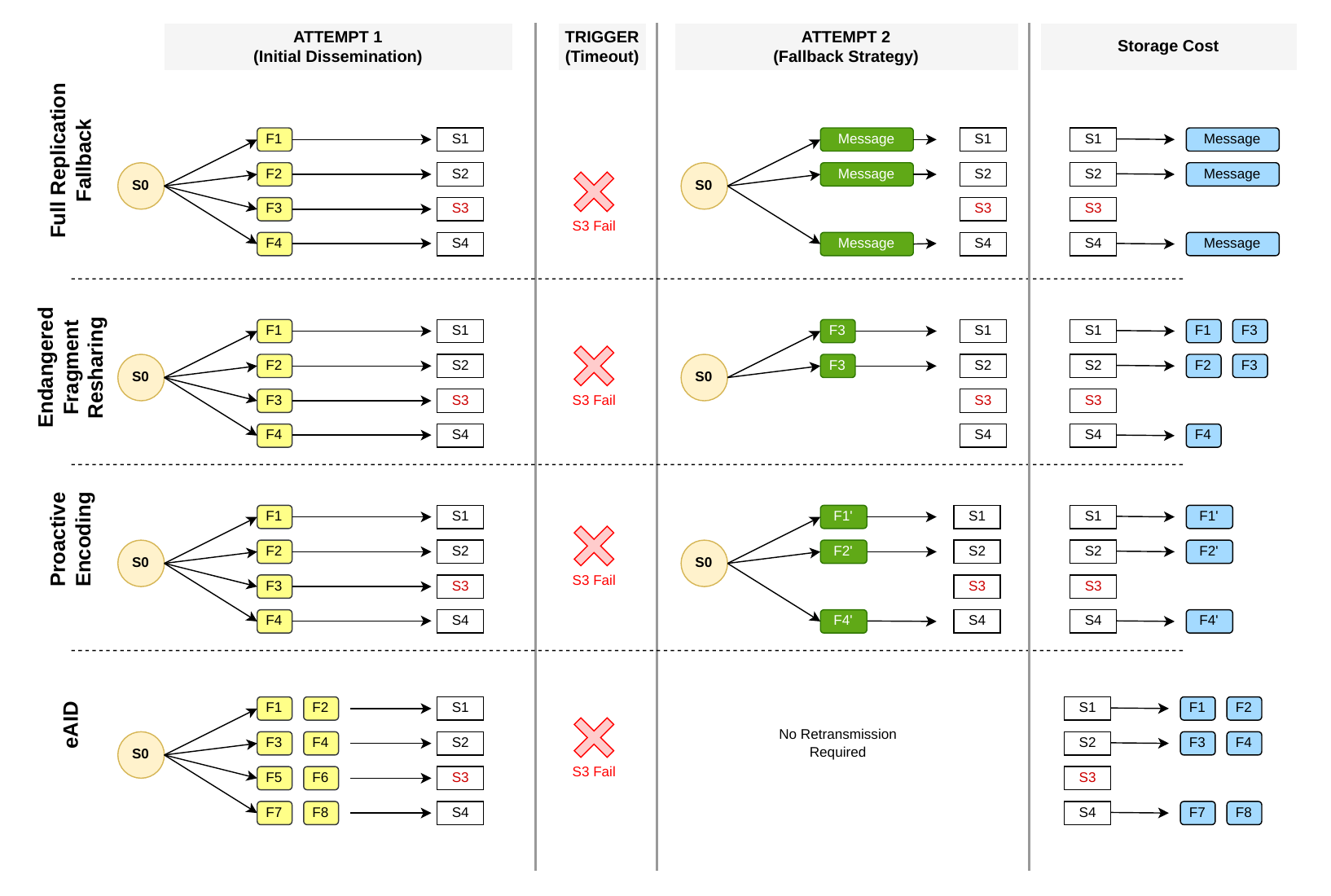} 
    \vspace{-1.5em} % 1. Reduces space between the image and the caption
    \caption{Scenario illustrating a single node failure (S3).
    In the Full Replication Fallback approach (top row), since the leader does not receive a response from S3, it reverts to full replication. In the Endangered Fragment Resharing approach (2nd row), the leader sends S3's fragment to $F=2$ randomly picked nodes. In the Proactive Encoding approach (3rd row), the leader re-encodes information using a $(F+1-1, F) = (2, 2)$ encoding scheme and disseminates new, larger fragments $\{F1', F2', F4'\}$ to the responsive nodes. 
    In \sys, the leader initially disseminates a tunable number ($2$ here) of distinct fragments to each node and no retransmission is required.
    }
    \label{fig:CodedRaft}
    \vspace{-1.75em} % 2. Reduces space between the caption and the text below it
\end{figure}

\subsection{Asynchronous Information Dispersal}
\label{sec:strawman}

Previous solutions deal with the key challenge of asynchrony in several ways. 
We first define the difference between a \textit{responsive}, \textit{unresponsive}, and \textit{failed} node as follows: a responsive node is a node that responds to messages in a timely manner, i.e., before the leader's timeout expires. An unresponsive node is a node that is still alive, but due to bad network conditions does not respond before the timeout expires. Notably, an unresponsive node can later become responsive if network conditions improve. Lastly, a failed node is a node that has crashed and will not respond to any incoming messages. The difficulty in asynchrony is that a leader cannot distinguish between an unresponsive and failed node as from its perspective they both behave the same. 

\textbf{Full Replication Fallback:} 
Figure \ref{fig:CodedRaft} (top row) demonstrates an approach where a leader starts the replication of every message by broadcasting a single encoded fragment to each node. If there are any unresponsive nodes, i.e., if the leader misses even a single acknowledgment,
the leader switches to full-copy replication. Note that in this case, a second round of retransmission is required and all nodes must store the entire message. 
Thus, with even a single unresponsive node, the latency doubles and the storage degenerates to full replication with $O(N)$ storage blowup. Some systems that implement this approach are \CRaft \cite{CRaft} and RS-Paxos ~\cite{RS-Paxos}.

% HRaft
% --------
\textbf{Endangered Fragment Resharing:}
Figure~\ref{fig:CodedRaft} (second row) illustrates an approach which starts the same way as Full Replication Fallback and addresses the abrupt performance degradation by reactively replenishing fragments across responsive nodes. Rather than defaulting immediately to full-copy replication,
the leader replicates fragments assigned to unresponsive nodes on available nodes. More specifically, if the leader misses acknowledgments from $f$ nodes, it sends copies of the fragments intended for the $f$ missing nodes to $F$ responsive nodes. Consequently, some nodes hold $f$ extra fragments. 
In many scenarios, resharing endangered fragments improves the storage required compared with falling back to full replication: rather than immediately degrading to full replication, it increases storage gradually, replicating only the fragments for which it misses acknowledgments.  
%As illustrated in Figure \ref{fig:CodedRaft}, Approach 2 also requires retransmission but only stores two fragments instead of the full data item.
However, if many nodes are unresponsive, the responsive nodes must store all the fragments designated for them. For example, if $N/4$ of the nodes are unresponsive, this approach incurs $O(N)$ storage blowup. Several systems ~\cite{HRaft, OceanStore, TotalRecall, Glacier, EternityService} adopt similar approaches. 

% FlexRaft
% -------------
\textbf{Proactive Encoding:} 
Another approach, which improves upon the two prior ones, proactively estimates the number of responsive nodes and adjusts the coding parameters based on this estimate.   
The leader adjusts the erasure code parameters and the assignment of fragments to nodes based on its current estimate. 
For example, if the leader estimates that one node may have failed, it adjusts the encoding scheme to $((F+1)-1, F)$. More generally, if the leader estimates that there are $f$ failures, it uses a $(F+1-f, F)$ encoding scheme. 
Accordingly, each fragment size increases from $\frac{B}{(F+1)}$ to $\frac{B}{(F+1-f)}$, and only $(F+1-f)$ fragments are needed to recover the encoded Message.
After re-encoding, the leader disperses the slightly larger fragment to responsive nodes. 
As shown in Figure \ref{fig:CodedRaft} (third row),  this approach may require a retransmission at the exact time a failure happens. Subsequent rounds do not require retransmissions as they make use of the existing coding parameters until the leader learns that the failure has been repaired. In general, this approach requires retransmission only if it receives fewer responses than it estimated using previous responses to dissemination. Although it lowers storage overhead, this approach ends up with a replicated log in which each entry may use a different coding scheme, complicating meta information maintenance and recovery. A Raft-based solution\cite{FlexRaft} adopts this approach.

\medskip

In summary, all the above approaches incur a retransmission round if they encounter unresponsive nodes; they complicate the reconstruction process because different messages may use different encodings; and they do not address post-dissemination storage maintenance. In the rest of the paper, we detail how \sys overcomes these shortcomings and present an evaluation of its performance.

\subsection{The Erasure Code}

The most common code used in Information Dispersal is Reed-Solomon codes \cite{Reed-Solomon}. However, a natural question is whether information dispersal could benefit from other codes such as LT codes~\cite{LubyTransform}, Fountain codes~\cite{Digital-Fountains}, Locally Recoverable Codes (LRC)~\cite{ErasureCodingAzure}, Minimum Storage Regenerating (MSR) codes~\cite{Clay-Codes} or Random Linear Network Codes (RLNC)~\cite{RLNC}.

MSR codes and LRC codes are primarily optimized for node repair. Since our approach is not concerned with repairing a node, we do not consider these classes of codes. Both LT and Fountain codes are classes of rateless codes that allow the encoder to generate an unlimited number of fragments. 
We deliberately chose a fixed-parameter Reed-Solomon scheme for four reasons.
First, Reed-Solomon is Maximum Distance Separable (MDS): \emph{any} $F+1$
fragments suffice to reconstruct $M$, deterministically. LT and Fountain
codes are only near-MDS and typically require $(1+\epsilon)(F+1)$ fragments
to decode with high probability; for the small dimensions typical in Raft
deployments (e.g., $F+1 = 3$ for $N = 5$), this overhead is both significant
and probabilistic, which would weaken the safety argument in
Lemma~\ref{lem:pruning-safety} that relies on exactly $F+1$ fragments always
being sufficient. RLNC achieves MDS behavior in expectation but only when
the random coefficient matrix is full-rank, which holds with high probability
over large fields but is not guaranteed. Second, Reed-Solomon fragments
carry only a small integer index as metadata, whereas RLNC fragments must
carry a length-$(F+1)$ coefficient vector and LT/Fountain fragments must
carry a degree and neighbor list. Avoiding per-fragment metadata is central
to \sys's goal of simplifying per-log-entry meta-information. Third,
autonomous pruning is sound under Reed-Solomon precisely because fragments
are interchangeable in count: a node can locally discard fragments knowing
that any $F+1$ surviving fragments cluster-wide will decode. With LT or
Fountain codes, the decodability of the surviving set depends on the
specific degree distribution of which fragments were retained, which would
require coordination to reason about safely. Fourth, although rateless codes
offer the advantage of generating fragments on demand, in practice this
benefit is modest in our setting: \sys's adaptive leader only generates
fragments per node based on its current
responsive quorum estimate. We elaborate on how many fragments the \sys leader generates in Section ~\ref{sec:dissemination-protocol}.
Beyond these design-level considerations, Reed-Solomon also offers a
practical implementation advantage: SIMD optimized libraries~\cite{ISA-L} (which we use in our implementation) deliver encoding and decoding throughput that, to the best of our knowledge,
is not matched by any production-grade rateless-code library.

Notwithstanding, \sys can work with any of these erasure codes. However due to the practicality of the optimized library, we choose Reed-Solomon codes.

%% file: problem.tex
\section{System Model}
\label{sec:problem}

\textbf{Asynchrony}
We will consider the classical \textit{asynchronous} model with $F < N/2$ benign failures. More specifically, the system consists of $N$ nodes interconnected via point-to-point message passing, such that $F < N/2$ nodes can crash. We will assume static membership such that once the protocol has started, no nodes will join or leave the system, besides the nodes that fail. Once a node fails, it will not recover and start participating in the protocol again. Furthermore, if a node fails, we make no assumptions about the messages it sent prior to failing, and we make no assumptions about whether messages arrive if sent by a node that eventually crashes. We will assume that links between non-faulty nodes are reliable, and that messages sent between non-faulty nodes eventually arrive, after a finite but unbounded delay. 

\begin{comment}
\Rithwik{Should we keep this here? I feel like it may make it confusing on whether our main solution is the information dispersal or its integration with a \KV}
\DM{Let's move this to the K-V section; here , we need a short definition of the IDA problem itself.} \Rithwik{I think we provide a definition of this in Section 2.1} \DM{True, but we need to repeat it here more formally, but without text motivating or explaining, just a definition. Here, we need to define correctness guarantees, e.g., when an IDA invocation is considered "complete"; later we \textbf{prove} correctness, here we state what is required. Here, give the definition as API, the same one used in the pseudo-code implementation.}
The problem we tackle is information dispersal integrated with a replicated \KV. The integration spans two layers. At the application layer, the \KV maintains a consistent view across nodes: every update is committed through a round of consensus before being applied. Under the hood, this consensus is provided by a log replication protocol, where a leader replicates a sequence of \LogEntries to followers. We integrate the information dispersal algorithm at this lower layer, replacing the full-payload entries that log replication would normally disseminate with erasure-coded fragments. Notably, nodes need not reconstruct an entry before persisting it. They only need to collect the information required to enable post-dissemination pruning.
\end{comment}
We formalize the information dispersal problem via the following API. 
\textsc{Disperse}$(M) \to \{\textsc{Done}\}$ is invoked at the leader 
to initiate dispersal of message $M$ across the $N$ nodes, returning 
\textsc{Done} when dispersal is complete. \textsc{Prune}$(\mathit{id})$ 
is invoked at any non-failed node to potentially discard locally held fragments 
of the message identified by $\mathit{id}$.

A dispersal protocol is \emph{correct} if it satisfies the following 
properties. \textbf{Dispersal Safety:} when \textsc{Disperse}$(M)$ 
returns \textsc{Done}, $M$ remains reconstructable from the fragments 
held across nodes, despite up to $F$ subsequent crash 
failures. \textbf{Pruning Safety:} after any invocation of 
\textsc{Prune}$(\mathit{id})$, $M$ remains reconstructable from the 
fragments held across nodes, despite up to $F$ subsequent 
crash failures. \textbf{Dispersal Termination:} every invocation of 
\textsc{Disperse}$(M)$ by a correct leader eventually returns 
\textsc{Done}, provided the leader does not fail.

\textbf{Performance measures.}
With these assumptions as our starting point, we shift our focus to the metrics we aim to optimize. Our primary goal is to enhance asynchronous performance by optimizing for two metrics:  \textbf{storage cost} and \textbf{dispersal latency}.

%The total storage cost is the total amount of storage nodes use in order to guarantee $M$ remains available; the per-node storage cost is the storage load inflicted per node, potentially non-uniformly.   

We define \textbf{storage cost} as the aggregate size of storage used across all nodes to persist the fragments associated with a message. For  a baseline,  full replication---where a message of size $B$ is sent to all $N$ nodes---incurs a total storage cost of $O(N \cdot B)$. We also define \textbf{per-node storage} as the storage required on a responsive node in the system.

We define \textbf{dispersal latency} as the time between the leader initiating dissemination for a message and the moment the leader is certain that the message has been persisted on a sufficient number of nodes to guarantee recovery. We focus on this metric to isolate the performance of the information dispersal protocol, explicitly decoupling it from the baseline overhead of the replicated \KV and client round-trip times.

%% file: protocol.tex
\section{\sys}
\label{sec:Solution}
\subsection{Dissemination and Reconstruction}
\label{sec:dissemination-protocol}
%\sys decouples dissemination from storage optimization. The process begins when the leader encodes a message $M$ using a $(F+1, (N-1)(F+1))$ erasure coding scheme. 

An \sys leader starts every dissemination attempt by encoding $M$ into a
\textbf{very large} number of fragments using a $(F+1, (F+1)\times(N-1))$
erasure coding scheme: creating $(F+1)\times N$ unique fragments such that
any $F+1$ can reconstruct $M$. The leader employs an adaptive dissemination
strategy, where it varies the number of fragments it sends to each node,
from 1 to $F+1$. Algorithm~\ref{alg:ida-dispersal} formalizes the leader side of this protocol, and
Algorithm~\ref{alg:ida-storage} specifies the corresponding storage node
behavior. The \sys leader uses previous responses to estimate a \RQ. Based on this estimate, the leader determines how many fragments
to assign to each node. For example, in steady state the leader may receive responses from all $N$ nodes in the previous dissemination and send only one fragment to each
node, similar to IDA. If the leader observes up to $N/4$ response omissions, it may increase the number of fragments per node by
one, and send at least two fragments per node. More generally, if the
leader estimates the \RQ to be $F + t$, it sends at least
$\lceil\frac{(F+1)}{t}\rceil$ distinct fragments to each node
(Algorithm~\ref{alg:ida-dispersal}, line 3 computing $f_{\text{per}}$). If
the leader collects responses from $F+t$ nodes acknowledging that they
received $\lceil\frac{(F+1)}{t}\rceil$ fragments each, there will always be
$F+1$ fragments present in the system to reconstruct the original message
$M$. If the leader does not collect $F+t$ acknowledgments, then it needs
to retransmit additional fragments, as handled by
line 16-20 in Algorithm~\ref{alg:ida-dispersal}.
\input{algorithms/leader-dissemination}
However, uniquely to \sys, the leader may prevent retransmission almost
always if it initially sends each node a slightly higher number of
fragments, say $\lceil\frac{(F+1)}{t'}\rceil$ where $t' < t$
(corresponding to the conservative margin $\delta = t - t'$ in
Algorithm~\ref{alg:ida-dispersal}). As an example, in an 11-node setup, the leader can initially disseminate two fragments to each node. If one node is unresponsive, the leader does not have to retransmit, and is able to tolerate up to $N/4 \approx 3$ unresponsive nodes without retransmission. Notably if all nodes do respond, the leader can later inform nodes to prune a fragment. 
There are two advantages in sending a higher number of fragments initially:
\begin{enumerate*}
    \item Lower the risk of retransmission (in the extreme case, zero risk when sending
    $F+1$ fragments per node).
    \item Faster completion of the dissemination phase, since a leader
    must wait for a lower number of responses.
\end{enumerate*}
If a higher number of nodes \textbf{do} respond, e.g., $F+t$ nodes
acknowledge that they store say $\lceil\frac{(F+1)}{t'}\rceil$ fragments
each, then the extra fragments will be immediately pruned via the
\textsc{Prune} procedure in Algorithm~\ref{alg:ida-storage}, which
each follower invokes opportunistically as the leader's threshold values
advance (see the pruning section below). Bandwidth is typically abundant,
while dispersal latency may be critical; hence, the leader can trade off
a slightly higher communication cost at the outset to avoid a latency
penalty later.
%\DM{I am debating whether the "marketting" points below should better move to the end of the Background section or stay here?}

Figure \ref{fig:CodedRaft} (bottom) illustrates \sys in the single-failure scenario we have used before. We illustrate the case where the leader initially assigns two fragments to each node. Even though there is a failure, due to its conservative approach, \sys does not need to retransmit fragments.

\subsection{Post Dissemination Pruning}
\label{sec:post-dissemination}
\sys's primary departure from prior works is that every node can prune the extra fragments post-dissemination by itself. Pruning does not require coordination or re-coding of information and can be carried autonomously for every node.

\begin{enumerate}
    \item \textbf{Asynchronous Collection:} After waiting for a certain number of responses in the initial dissemination round, the leader continues to collect responses for $M$ in the background.
    
    \item \textbf{\RQcapital Updates:} As the leader gathers more responses after initial dissemination completes, it updates its estimate of the number of nodes that have stored fragments of $M$ and sends this information to nodes in the system. 
    
    \item \textbf{Dynamic Storage Optimization:} Updates about the \RQ permit nodes to dynamically discard their fragments. For instance, if a node learns that $\frac{3N}{4}$ are storing fragments of $M$, it may safely discard all but two of its local fragments. In general, if a node determines that $F + t$ nodes are storing fragments, it only needs to retain $\lceil\frac{(F+1)}{t} \rceil$ fragments. Since there are no overlaps between the fragments that nodes store, it does not matter which fragments the node drops. 
\end{enumerate}
Through a combination of these three approaches, \sys uniquely adapts to unstable network conditions by:
\begin{enumerate*}
    \item Only varying the number of fragments assigned per node;
    \item Allowing different numbers of fragments to be stored by different nodes: so long as $F+1$ fragments can be collected in the entire system, $M$ will remain available;
    \item Not requiring re-encoding of $M$;
    \item Not impacting the recovery computation, which always takes $F+1$ fragments and reconstructs $M$ in the same manner, oblivious to the assignment of fragments to nodes
    \item Not requiring maintaining meta-information per message $M$; 
    \item Almost always avoiding retransmission rounds.
\end{enumerate*}

Beyond quantifiable storage and latency gains, a primary design objective is protocol simplicity. Existing solutions often compromise this goal, suffering from either storage inefficiency or significant metadata overhead. By using a static encoding scheme with uniformly distributed fragments, reconstruction is reduced to a simple collection of any $F+1$ fragments, eliminating the need to track fragment placement or manage evolving encoding states.
%This approach differs from traditional IDA, where storage optimization is on the critical path. The primary benefit is that storage reduction happens asynchronously. In periods of network asynchrony, the leader does not need to stall or re-transmit fragments if it does not receive a sufficient number of responses. Once it has collected the $F+1$ threshold, it can commit and proceed to the next message.

\input{algorithms/follower}

\subsection{Safety and Liveness}
\label{sec:proof}
\begin{lemma}[Dispersal Safety]
\label{lem:dispersal-safety}
\textit{When \textsc{Disperse}$(M)$ returns \textsc{Done}, $M$ is
reconstructable despite up to $F$ failures.}
\end{lemma}

\begin{proof}
Let $A := A[M.\mathit{id}]$ at the moment \textsc{Disperse} exits its
main loop. The exit condition (line~\ref{ln:exit}) guarantees
$|A| \geq F + t'$, and membership in $A$ requires
$\mathit{held}[M.\mathit{id}, j] \geq f_{\text{per}}$ for every
$j \in A$ (lines~\ref{ln:ack-add} and~\ref{ln:commit-bounds}). Since
each $S_j$ is a disjoint block of the partition
$\mathcal{F} = S_1 \cup \cdots \cup S_N$, fragments held by distinct
nodes in $A$ are disjoint.

We first verify the invariant $f_{\text{per}} \cdot t' \geq F + 1$ at
every point where the leader checks the exit condition. Initially
(line~\ref{ln:encode-init}), $f_{\text{per}} = \lceil (F+1)/t' \rceil$
with $t' = \max(1, t - \delta) \geq 1$, so
$f_{\text{per}} \cdot t' \geq F + 1$ by definition of the ceiling.
After a timeout the leader recomputes
$t'_{\text{new}} = \max(1, |A| - F)$ and
$f_{\text{new}} = \lceil (F+1)/t'_{\text{new}} \rceil$
(line~\ref{ln:recompute-bounds}) and commits these to $t'$ and
$f_{\text{per}}$ (line~\ref{ln:commit-bounds}), preserving the
invariant.

After up to $F$ failures, at least $|A| - F \geq t'$ nodes from $A$
survive. Each surviving node holds at least $f_{\text{per}}$ fragments,
and these fragments are pairwise disjoint across nodes. The total
number of surviving fragments is therefore
\[
S \;\geq\; (|A| - F) \cdot f_{\text{per}}
  \;\geq\; t' \cdot f_{\text{per}}
  \;\geq\; F + 1,
\]
which suffices to reconstruct $M$.
\end{proof}

\begin{lemma}[Pruning Safety]
\label{lem:pruning-safety}
\textit{If every node prunes its fragments according to the rule}
\[
r(|Q|) \;=\; \left\lceil \frac{F+1}{|Q| - F} \right\rceil,
\]
\textit{then $M$ remains reconstructable after any pattern of up to
$F$ failures, regardless of how nodes apply the rule using their local
views of $|Q|$.}
\end{lemma}

\begin{proof}
Let $Q_{\mathit{hi}}$ be the set of nodes corresponding to the highest acknowledgment count, $|Q_{\mathit{hi}}|$, known to any individual node in the system.

Every surviving node $u$ that has pruned its fragments did so autonomously based on some local view $|Q_u| \le |Q_{\mathit{hi}}|$, retaining $r(|Q_u|)$ fragments (lines~\ref{ln:prune-rule}--\ref{ln:prune-discard}). Because the pruning function $r(x) = \lceil \frac{F+1}{x - F} \rceil$ is monotonically non-increasing as $x$ increases, a lower local view of the quorum results in a node retaining more fragments. Therefore, each surviving node retains at least the minimum required by the highest view: $r(|Q_{\mathit{hi}}|) = \lceil \frac{F+1}{|Q_{\mathit{hi}}| - F} \rceil$ fragments.

After up to $F$ worst-case failures, at least $|Q_{\mathit{hi}}| - F$ nodes from the set $Q_{\mathit{hi}}$ will survive. Since the initial dispersal guarantees that fragments held by different nodes are strictly disjoint, the total number of surviving fragments, $S$, is bounded below by the number of surviving nodes multiplied by the minimum number of fragments each retained:
\[
S \;\ge\; (|Q_{\mathit{hi}}| - F) \cdot
\left\lceil \frac{F+1}{|Q_{\mathit{hi}}| - F} \right\rceil
\;\ge\; F + 1.
\]
Because $S \ge F+1$, the threshold to decode the $(F+1, (F+1)\times(N-1))$ Reed-Solomon scheme is always met. Thus, $M$ remains reconstructable under pruning.
\end{proof}

\begin{lemma}[IDA Termination]
\label{lem:ida-termination}
Every invocation of \textsc{Disperse}$(M)$ by a correct leader eventually
returns \textsc{Done}, provided the leader remains correct throughout.
\end{lemma}
\begin{proof}
In round $k$, let $f_{\text{per}}^{(k)}$ be the per-node fragment count and
$t'^{(k)}$ the wait threshold. If round $k$ collects $F + t'^{(k)}$
acknowledgments, \textsc{Disperse} returns by
Lemma~\ref{lem:dispersal-safety}. Otherwise, a timeout occurs, and the
leader recalculates the wait threshold as
$t'^{(k+1)} = \max(1, |A^{(k)}| - F)$ and the new fragment requirement
$f_{\text{per}}^{(k+1)} = \lceil (F+1)/t'^{(k+1)} \rceil$
(line~\ref{ln:recompute-bounds}). The leader then broadcasts additional
fragments to \emph{all} nodes $j$ with
$f_{\text{new}} > \mathit{sent}[M.\mathit{id}, j]$
(line~\ref{ln:retransmit}), rather than restricting retransmissions to
any earlier subset.

Because retransmissions are broadcast globally and each node's highest
acknowledged fragment count $\mathit{held}[M.\mathit{id}, j]$ is only
ever increased (line~\ref{ln:held-update}), never discarded across
rounds, the total number of fragments successfully held by correct
nodes is monotonically non-decreasing in $k$. The wait threshold $t'$
is bounded below by $1$, which upper-bounds the per-node fragment count
at $f_{\text{per}} = F+1$.

Eventually $t' = 1$, in which the termination condition reduces to
$|A[M.\mathit{id}]| \ge F+1$ and each node is required to hold all
$F+1$ fragments of its assigned set $S_j$. By assumption, there are at
least $F+1$ responsive nodes. Each such node, upon eventually receiving
its $F+1$ fragments, replies with an acknowledgment $q = F+1$, which
the leader eventually delivers. Thus $|A[M.\mathit{id}]|$ eventually
reaches $F+1$ and the outer \textbf{while} loop exits
(line~\ref{ln:exit}), so \textsc{Disperse} returns \textsc{Done}.
\end{proof}

%% file: algorithms/leader-dissemination.tex
\begin{algorithm}[!htbp]
\caption{\sys IDA Dispersal (Leader)}
\label{alg:ida-dispersal}
\footnotesize
\begin{algorithmic}[1]
\State \textbf{State:} $\mathit{responsiveQuorum} \gets N$;
       $A[\cdot]$: ack sets per id;
       $\mathit{sent}[\mathit{id}, j], \mathit{held}[\mathit{id}, j]$:
       \# fragments sent to / acked by $j$;
       $f_{\text{per}}$: current per-node fragment threshold
\Statex
\Procedure{Disperse}{$M$}
    \State $\mathcal{F} \gets$ \Call{Encode}{$M$};
           $\;t \gets \mathit{responsiveQuorum} - F$;
           $\;\boldsymbol{t' \gets \max(1, t - \delta)}$;
           $\;f_{\text{per}} \gets \lceil (F+1)/t' \rceil$ \label{ln:encode-init}
    \State Partition $\mathcal{F}$ into $S_1, \ldots, S_N$ with $|S_j| = F+1$; \Comment{Each node gets a pool of $F{+}1$ unique fragments}
           $\;A[M.\mathit{id}] \gets \emptyset$
    \For{each node $j$}
        \State Send \textsc{Message}$(M.\mathit{id}, S_j[1 \dots f_{\text{per}}])$ to $j$;
               $\;\mathit{sent}[M.\mathit{id}, j] \gets f_{\text{per}}$;
               $\;\mathit{held}[M.\mathit{id}, j] \gets 0$
    \EndFor
    \While{$|A[M.\mathit{id}]| < F + t'$}
        \Comment{Continuously loop and adjust the number of fragments sent}
        \State Start timer $\tau$
        \While{$|A[M.\mathit{id}]| < F + t'$ \textbf{and} timer not expired}
            \State Wait for \textsc{Ack}$(\mathit{id}, j, q)$;
                   \textbf{if} $q > \mathit{held}[M.\mathit{id}, j]$
                   \textbf{then} $\mathit{held}[M.\mathit{id}, j] \gets q$ \label{ln:held-update}
            \State \textbf{if} $\mathit{held}[M.\mathit{id}, j] \geq f_{\text{per}}$
                   \textbf{then} $A[M.\mathit{id}] \gets A[M.\mathit{id}] \cup \{j\}$ \label{ln:ack-add}
        \EndWhile
        \If{$|A[M.\mathit{id}]| \geq F + t'$} \label{ln:exit} \textbf{break} \EndIf
        \Statex \hspace{1.5em} \Comment{timeout: recalculate bounds and send additional fragments}
        \State $\boldsymbol{t'_{\text{new}} \gets \max(1, |A[M.\mathit{id}]| - F)}$;
               $\;f_{\text{new}} \gets \lceil (F+1)/t'_{\text{new}} \rceil$ \label{ln:recompute-bounds}
        \For{each node $j$ \textbf{with} $f_{\text{new}} > \mathit{sent}[M.\mathit{id}, j]$}
            \State Send \textsc{Message}$(M.\mathit{id}, S_j[\mathit{sent}[M.\mathit{id}, j]{+}1 \dots f_{\text{new}}])$ to $j$;
                   $\;\mathit{sent}[M.\mathit{id}, j] \gets f_{\text{new}}$ \label{ln:retransmit}
        \EndFor
        \State $t' \gets t'_{\text{new}}$;
               $\;f_{\text{per}} \gets f_{\text{new}}$;
               $\;A[M.\mathit{id}] \gets \{ j \mid \mathit{held}[M.\mathit{id}, j] \geq f_{\text{per}} \}$ \label{ln:commit-bounds} \Comment{Recompute Ack set under new threshold}
    \EndWhile
    \State $\mathit{responsiveQuorum} \gets |A[M.\mathit{id}]|$;
           $\;$\textbf{return} \textsc{Done}
\EndProcedure
\Statex
\Procedure{OnLateAck}{$\mathit{id}$, $j$, $q$}
    \Comment{after \textsc{Disperse} returns}
    \If{$q > \mathit{held}[\mathit{id}, j]$}
        \State $\mathit{held}[\mathit{id}, j] \gets q$;
               \textbf{if} $j \notin A[\mathit{id}]$ \textbf{and} $\mathit{held}[\mathit{id}, j] \geq f_{\text{per}}$
               \textbf{then} $A[\mathit{id}] \gets A[\mathit{id}] \cup \{j\}$
        \State Send \textsc{AckUpdate}$(\mathit{id}, |A[\mathit{id}]|)$ to all nodes \Comment{Update nodes so they can prune}
    \EndIf
\EndProcedure
\end{algorithmic}
\end{algorithm}

%% file: algorithms/follower.tex
\begin{algorithm}[!htbp]
\caption{\sys: Storage and Autonomous Pruning}
\label{alg:ida-storage}
\footnotesize
\begin{algorithmic}[1]
\State \textbf{State:} $\mathit{store}[\mathit{id}]$: fragments held for message id;
       $Q[\mathit{id}]$: highest known \# nodes holding id
\Statex
\Procedure{OnMessage}{$\mathit{leader}, \mathit{id}, \mathit{fragments}$}
    \Comment{from leader's \textsc{Message} in Alg.~\ref{alg:ida-dispersal}}
    \State $\mathit{store}[\mathit{id}] \gets \mathit{store}[\mathit{id}] \cup \mathit{fragments}$
    \State Send \textsc{Ack}$(\mathit{id}, \mathit{self}, |\mathit{store}[\mathit{id}]|)$ to leader
           \Comment{ack carries current fragment count}
\EndProcedure
\Statex
\Procedure{OnAckUpdate}{$\mathit{id}, q$}
    \Comment{from leader's \textsc{AckUpdate} in Alg.~\ref{alg:ida-dispersal}}
    \If{$q > Q[\mathit{id}]$}
        \State $Q[\mathit{id}] \gets q$; \Call{Prune}{$\mathit{id}$}
    \EndIf
\EndProcedure
\Statex
\Procedure{Prune}{$\mathit{id}$}
    \State $q \gets Q[\mathit{id}]$;
           $\;r \gets \lceil (F+1)/(q - F) \rceil$ \label{ln:prune-rule}
           \Comment{required fragments per node}
    \While{$|\mathit{store}[\mathit{id}]| > r$} \label{ln:prune-discard}
        \State Discard any one fragment from $\mathit{store}[\mathit{id}]$
    \EndWhile
\EndProcedure
\end{algorithmic}
\end{algorithm}

%% file: useCase.tex
\section{Use Case: Encoded information dispersal in a \KV}
\label{sec:use-case}
To demonstrate the benefit of the \sys approach, we build an application using it. We choose a replicated \KV backed by log 
replication for two reasons: they are a demanding target where both 
dispersal latency and long-term storage matter, and they are the setting 
in which prior IDA-integrated protocols~\cite{HRaft, CRaft, FlexRaft} 
have been evaluated, enabling direct comparison.

Replicated \KVs are the storage substrate behind 
databases~\cite{Spanner}, distributed file systems~\cite{pan2021facebook}, and 
control planes~\cite{Brooker2020}: a replicated log establishes the 
authoritative order of operations across nodes, and a \KV applied on 
top serves the live data. With IoT~\cite{Microsoft_IoTSignals_2019} and AI 
workloads~\cite{JLL_AI_DataCenters_2024} pushing enterprise storage demand into the 
zettabyte\footnote{a zettabyte is $10^9$ terabytes} range~\cite{IDC_DataSphere}, 
the on-disk footprint of these stores has become an increasing cost. 
Bandwidth spent during short-term dissemination is transient, while 
storage is paid for indefinitely, on every replica. Key-value stores 
are therefore a setting where \sys's post-dissemination pruning 
translates most directly into long-term operational savings, making 
them a natural target for evaluating the protocol's storage benefits 
in practice.

Integrating \sys into this setting requires starting from a log-replication core that consistently sequences updates to key-value entries across replicas, and replacing its internal information-dispersal mechanism with \sys-based erasure-coded dissemination. To build this implementation, which we call \sysConsensus, we borrow code from an existing replicated \KV system~\cite{FlexRaft} and modify its dissemination layer accordingly.

In Section~\ref{sec:evaluation}, we evaluate \sysConsensus against several prior protocols that integrate IDA with log replication in the same setting~\cite{FlexRaft, HRaft, CRaft}. 
%\DM{add others?}.

%A leader can employ IDA to disseminate proposals to save communication and storage costs, but doing so poses operational challenges. The leader must successfully disseminate fragments to enough nodes to guarantee data availability, yet cannot wait indefinitely for responses since these protocols operate under asynchronous network conditions. Once the leader gives up waiting, it must adjust the coding parameters, the fragment-to-node assignment, or both.

%% file: implementation.tex
\section{Implementation}
\label{sec:implementation}

We implemented \sysConsensus in C++ building upon the existing \FlexRaft implementation~\cite{FlexRaft-Code} which contains approximately 6,000 lines of code and modified it by adding roughly 1,000 lines of code. We choose FlexRaft as a starting point because it too implements a KV store backed by log replication. To facilitate evaluation, we constructed a distributed key-value store backed by RocksDB v7.3.1 \cite{RocksDB}. Architecturally, each node hosts an \sysConsensus alongside the RocksDB state machine, utilizing a background thread to asynchronously apply committed entries. We provide further details of the integration of \sys in the Appendix \ref{sec:kv-integration}.

%\begin{itemize}

%\item \textbf{Transition to two fragments (at Threshold~1):} For entries indexed up to \textbf{Threshold~1} (acknowledged by $3/4N$ nodes), the follower node retains two fragments. This is necessary because, with $F < N/2$ potential failures, even in the worst case of $F$ failures, there will still be $N/4$ nodes that haven't failed. To ensure the cluster retains enough fragments to reconstruct the entry (requiring $F+1$ fragments), each node must store two fragments ($2 \times N/4 \ge F+1$).
%\item \textbf{Transition to one fragment (at Threshold~2):} For entries indexed up to \textbf{Threshold~2} (acknowledged by $N$ nodes), the follower node retains only a single fragment. Since the entry is replicated to every node, even in the worst-case scenario of $F$ failures, the $F+1$ surviving nodes guarantee the availability of $F+1$ unique fragments, satisfying the reconstruction requirement.
%\end{itemize}

Our implementation utilizes a leader-driven dissemination model, where the leader broadcasts network response information to all followers. These metrics are derived from data gathered via periodic heartbeats and standard message exchanges in Raft. While we focus on this leader-centric approach, alternative  mechanisms such as all-to-all peer communication could be employed to allow nodes to independently assess network state. Such decentralized state acquisition would enable nodes to perform local storage optimizations without relying on the leader's view, a direction we reserve for future work. We provide pseudocode for \sys's integration with Log Replication in Appendix \ref{sec:RaftureRaftIntegration}.

%% file: evaluation.tex
\section{Evaluation}
\label{sec:evaluation}
%\DM{Change to: We provide two sets of evaluations, one analytical and experimental. The anlaytical provides predicted storage of Craft, Hraft, FlexRaft and \sys against unstable network condition (that may be difficult to emulate in live settings). The experimental evaluation compares an implementation of \sys forked from FlextRaft against FlexRaft itself, thus elimintating any side-effects of different communication libraries and the like. We chose FlexRaft as state of art .... blah blah.}

We evaluate \sys both as a standalone dispersal protocol (analytically) and as the engine of a key-value store, \sysConsensus (experimentally). Analytically, we study the dispersal protocol in isolation to separate its behavior from any application above it and to reproduce network-partition scenarios that are hard to stage in a live cluster. We compare \sys against three dispersal substrates: Full Replication Fallback, Endangered Fragment Resharing, and Proactive Encoding. Experimentally, we evaluate \sysConsensus against \FlexRaft, the only system that adopts the Proactive Encoding approach. Its codebase~\cite{FlexRaft-Code} provides a robust open-source implementation of this dispersal substrate integrated with a \KV. For the experimental study, we build \sysConsensus by modifying \FlexRaft's dispersal layer to align with \sys, and use the unmodified \FlexRaft as our baseline. This ensures any observed performance differences can be attributed to the dispersal protocol itself.

Through our evaluation we aim to demonstrate the following claims: \textbf{C1} \sys and \sysConsensus maintain the same storage utilization as \FlexRaftIDA and \FlexRaft, respectively, under normal-case network conditions; \textbf{C2} \sys and \sysConsensus maintain better storage utilization than \FlexRaftIDA and \FlexRaft, respectively, under network partitions; \textbf{C3} \sys maintains better dispersal latency than all previous approaches under unstable network conditions; \textbf{C4} \sysConsensus maintains the same or better performance than \FlexRaft in a distributed setup. 

\subsection{Experimental Setup}
We conducted our experiments on a high-performance DigitalOcean Droplet equipped with 16 vCPUs and 64GB of RAM. We primarily focus on a cluster within a single physical machine to isolate the storage metrics from external network variance. Since our primary evaluation metric is storage utilization efficiency—which is a function of the consensus state rather than packet latency—colocating the nodes allows us to deterministically reproduce partition scenarios without the noise of non-deterministic network jitter. We also include performance evaluation of a distributed setup where we deploy five nodes with the same configuration in DigitalOcean's SFO3 region with a colocated client sending transactions. We include the results of this evaluation in Section \ref{sec:performance}

\subsection{Storage Utilization}

We first present an analytical evaluation of \sys, and verify our results with an experimental evaluation of \sysConsensus. 
Figure~\ref{fig:partition:two} depicts storage utilization when two nodes are partitioned from the network for increasing cluster sizes. The two nodes are partitioned until 2000 entries and don't respond to any incoming messages. Once the partition ends the nodes start responding to messages. Under these conditions, \sys performs as well as or slightly worse than \FlexRaftIDA until the partition ends. Once the partition ends, \sys has better performance as it is able to adapt and reduce the higher storage cost required during the partition. Initially, upon detecting an unresponsive node, \FlexRaftIDA interprets this as a node crash and adjusts its encoding strategy accordingly. In contrast, \sys assumes the node is simply unresponsive (asynchrony). By assuming the node is unresponsive, \sys incurs a higher storage overhead as it waits for the node to become responsive. Initially, this is a slight disadvantage as all nodes in \sys are storing more information than required. However, once the node rejoins the network, nodes can prune the extra storage incurred and reduce storage down to the optimal amount.

%This performance gap stems from how they interpret an unresponsive node: FlexRaft interprets the lack of response as a node crash, allowing it to immediately adjust its encoding strategy. In contrast, \sysConsensus assumes the node is simply unresponsive (asynchrony). Consequently, \sysConsensus continues to operate without modifying its encoding scheme to exclude the missing node, incurring a higher storage overhead to maintain safety assumptions while waiting for the node to rejoin the network.

\begin{figure*}[t]
    \centering
    \vspace{-15pt}
    
    \begin{subfigure}[t]{0.24\linewidth} % Increased to 0.33
        \centering
        \includegraphics[width=\linewidth]{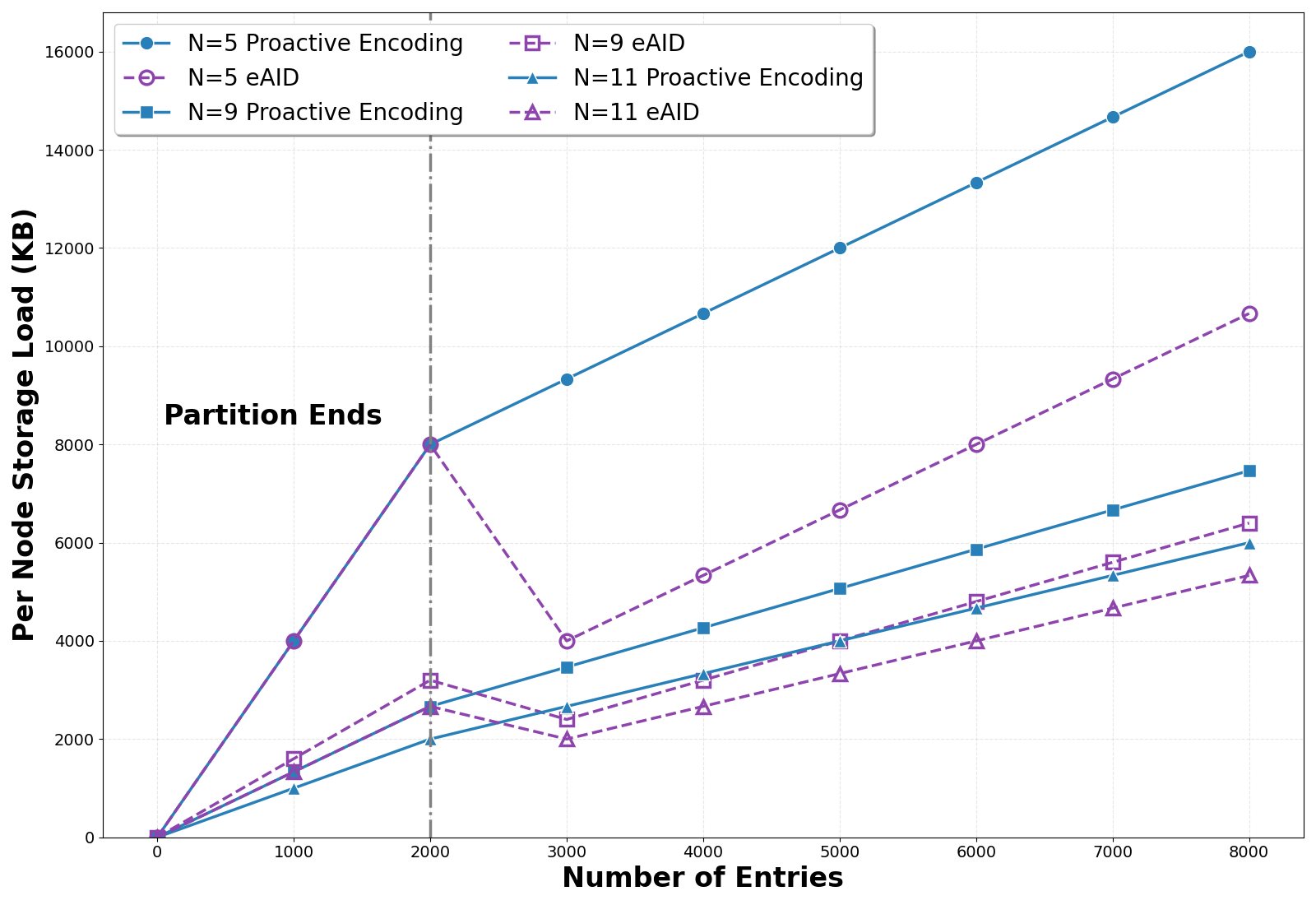}
        \vspace{-14pt}
        \caption{Two partitioned nodes}
        \label{fig:partition:two}
    \end{subfigure}% <--- VITAL: This % removes the space after the image
    \hfill% <--- Pushes images apart using only the remaining 1% of space
    \begin{subfigure}[t]{0.24\linewidth} 
        \centering
        \includegraphics[width=\linewidth]{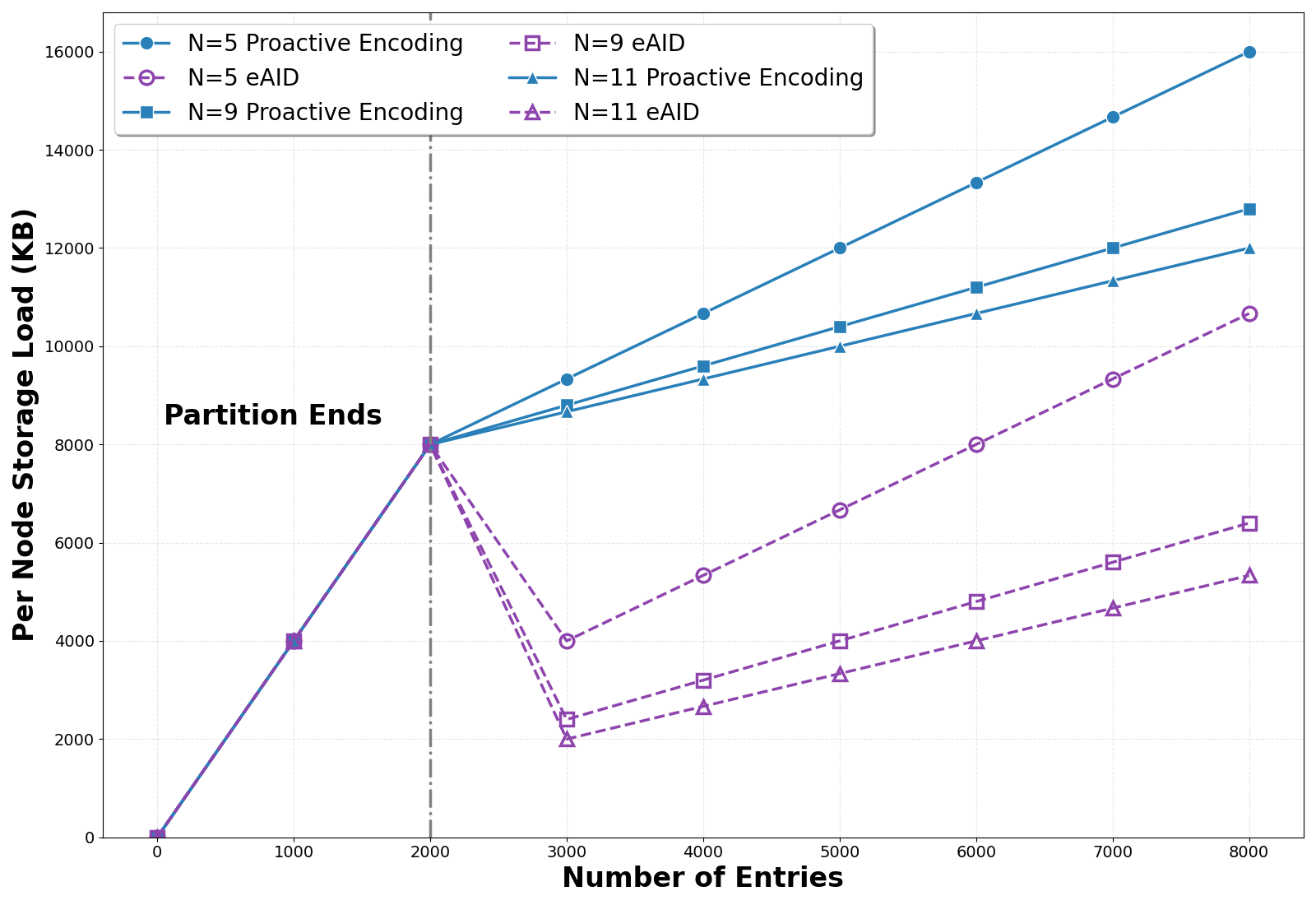}
        \vspace{-14pt}
        \caption{$F$ partitioned nodes}
        \label{fig:partition:max}
    \end{subfigure}% <--- VITAL: This % removes the space
    \hfill%
    \begin{subfigure}[t]{0.24\linewidth}
        \centering
        \includegraphics[width=\linewidth]{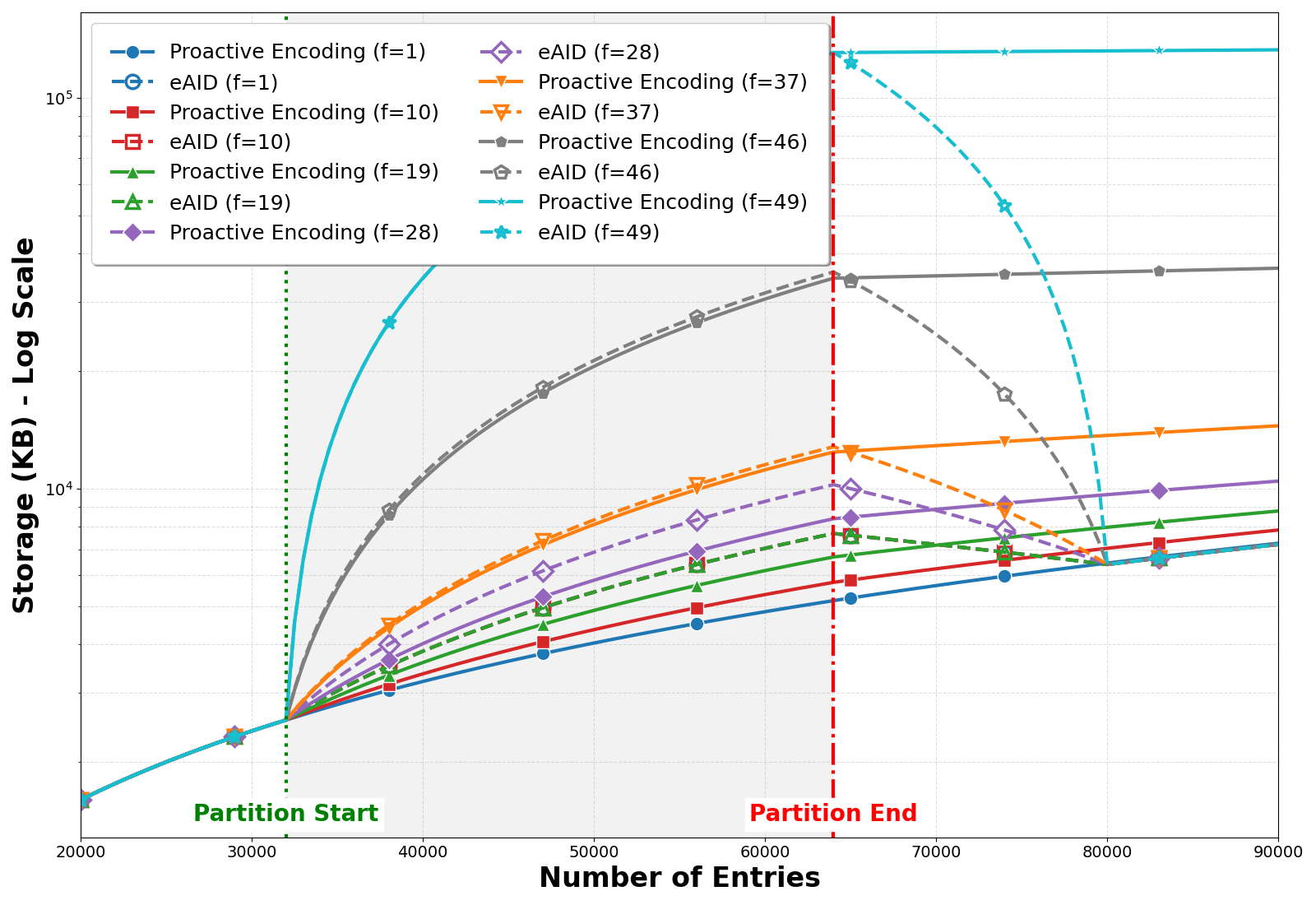}
        \caption{Storage cost; $N=99$ and $f$ partitioned nodes}
        \label{fig:increasingF}
    \end{subfigure}
\hfill%
    \begin{subfigure}[t]{0.24\linewidth} 
        \centering
        \includegraphics[width=\linewidth]{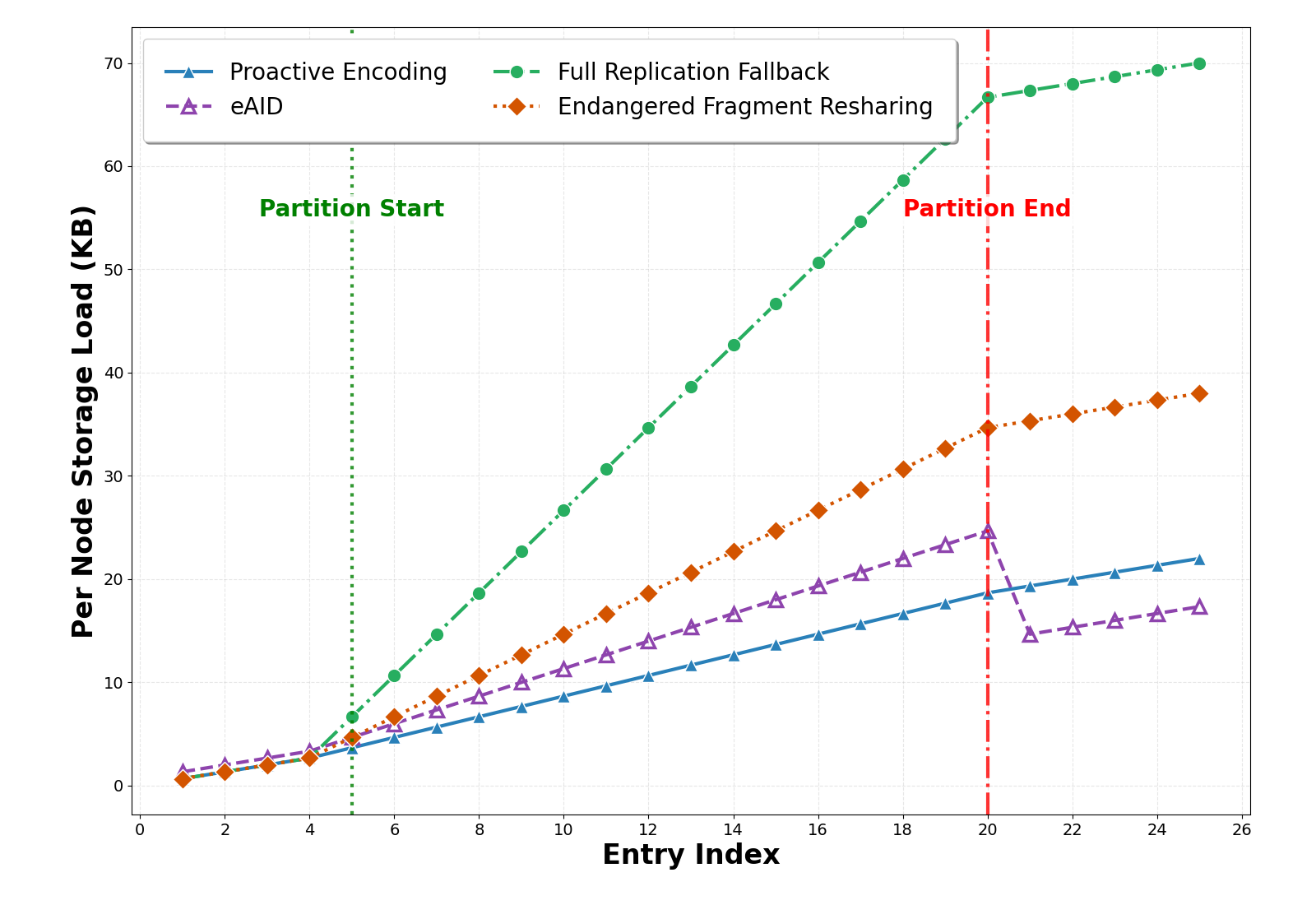}
        \vspace{-14pt}
        \caption{Zoom on partition start/end; $N=11$, $f=2$}
        \label{fig:StorageCost}
    \end{subfigure}

    \vspace{-10pt}
    \caption{Analytical evaluation of storage utilization under network partitions.}
    \label{fig:partition}
    \vspace{-2ex}
\end{figure*}

\begin{figure}[t]
\centering % Centers the entire group of subfigures
\begin{subfigure}{0.33\linewidth}
    \centering
    \includegraphics[width=\linewidth, clip]{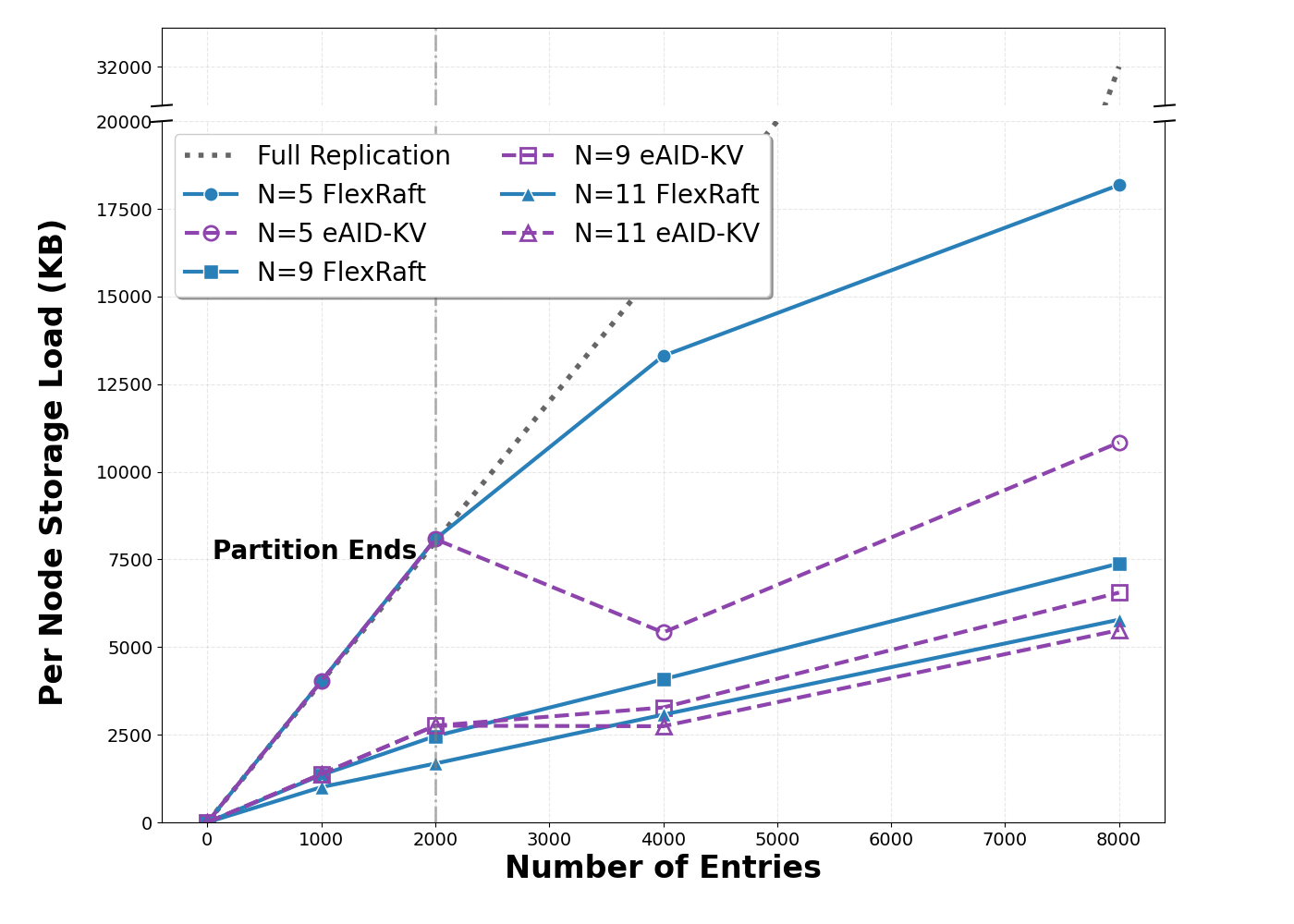}
    \vspace{-14pt}
    \caption{Two partitioned nodes}
    \label{fig:serverpartition:two}
\end{subfigure}%
\hspace{1em} % Adds a small, fixed gap between the figures (adjust 1em as needed)
\begin{subfigure}{0.33\linewidth}
    \centering
    \includegraphics[width=\linewidth, clip]{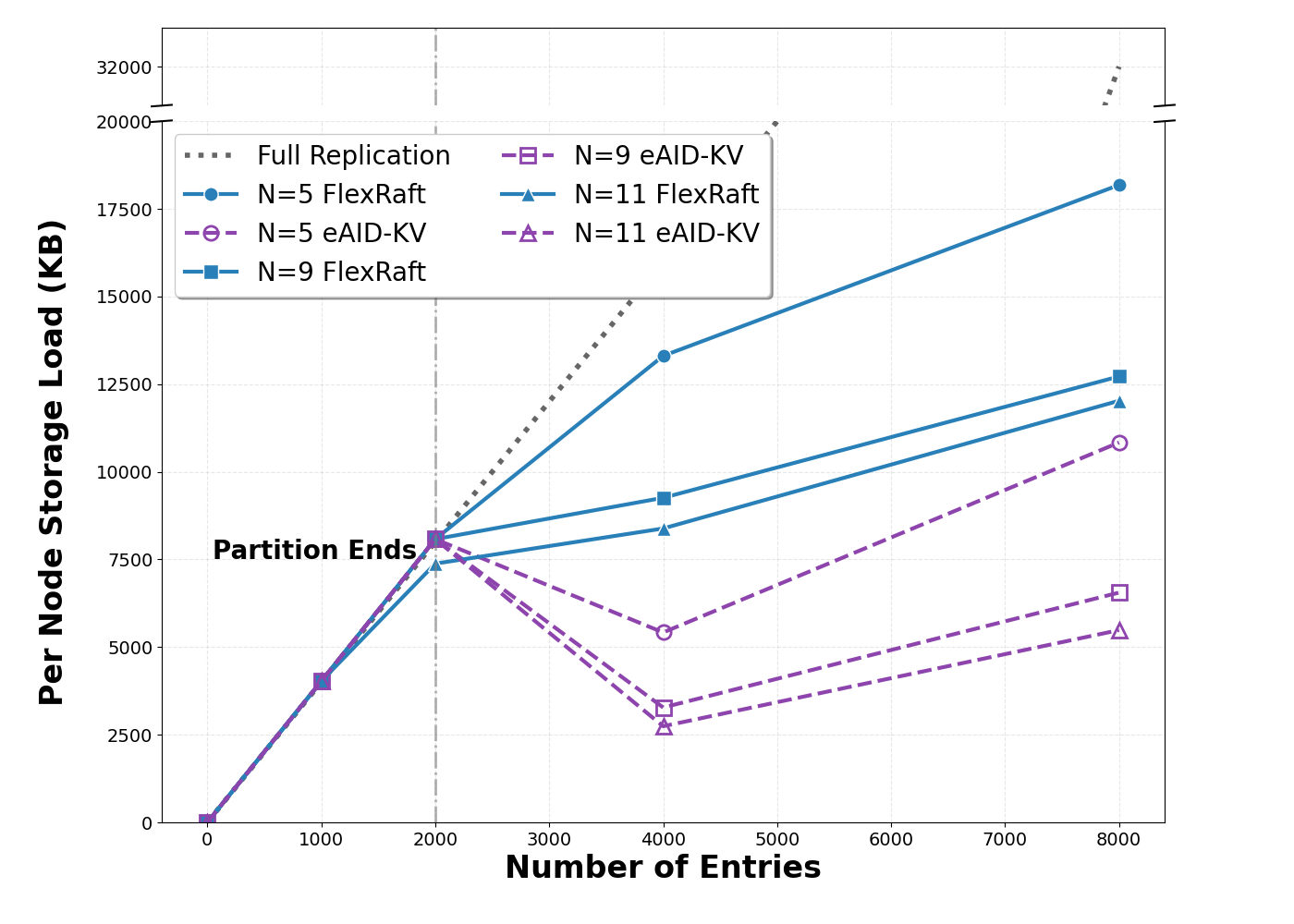}
    \vspace{-14pt}
    \caption{$F$ partitioned nodes}
    \label{fig:serverpartition:max}
\end{subfigure}

\vspace{-10pt}
\caption{Experimental evaluation of storage utilization for \sysConsensus and \FlexRaft under a network partition that recovers at 2000 entries. We also show the storage utilization of Full Replication.}
\label{fig:serverpartition}
\vspace{-3ex}
\end{figure}
Figure \ref{fig:partition:max} presents a similar scenario where $F$ nodes are initially partitioned and rejoin the cluster after 2000 entries. Initially for entries $0-2000$, \sys exhibits the same storage utilization as \FlexRaftIDA as both systems need to store the entire entry to maintain data availability. With only $F+1$ responsive nodes remaining, \sys ensures data availability by storing $F+1$ fragments per entry. However, the system's ability to recover storage is demonstrated once the partition heals. When the leader confirms the $F$ missing nodes have returned and are responsive, \sys performs a retrospective optimization. The storage utilization sharply drops (visible at the $2000\rightarrow3000$ entry mark) as the system reduces the redundancy of prior entries from $F+1$ fragments down to one fragment per node.  

 Figure~\ref{fig:increasingF} extends the analysis to $N=99$ subject to an increasing number of partitioned nodes. A partition occurs between $32{,}000$ entries and $64{,}000$ entries. We observe that the same pattern holds with \sys consistently recovering to optimal storage post-partition 
while \FlexRaftIDA does not. The distinction is that \FlexRaftIDA 
cannot retroactively re-optimize storage after a partition without 
reconstruction and re-encoding, while \sys leverages the distinction 
between crashed and unresponsive nodes to recover the temporary 
overhead it incurs without any re-encoding or reconstruction.

Figure \ref{fig:StorageCost} shows a zoomed in version of Figure \ref{fig:partition:two} utilizing a configuration of $N=11$ and $f=2$. By zooming into the graph we can see the slightly higher storage cost of \sys at the beginning, but see the same benefits of its post commit. This confirms Claim \textbf{C2}: while \sys may temporarily use more storage to handle partitions, it uniquely recovers that space and maintains better storage utilization than \FlexRaftIDA once the cluster stabilizes. 

We  verify that these analytical storage predictions hold for the deployed \sysConsensus system. Figure~\ref{fig:serverpartition} reports the corresponding empirical measurements from \sysConsensus deployed against \FlexRaft, under the same two-node and $F$-node partition scenarios. The empirical curves track the analytical predictions for \sys, confirming that the dispersal behavior survives integration into the \KV.

\subsection{Dispersal Latency of \sys (Analytical)}

Figure \ref{fig:networkHops} illustrates the simulated latency required for each protocol to output \textsc{Done} on an
entry under unstable network conditions. We define latency as the total time it takes from when a leader initiates dissemination, to when it outputs done, excluding leader-election round trips and assuming a steady state leader. To model network instability, we draw the nodes response time from a normal distribution centered at 0.8~ms with a standard deviation of 0.15~ms. We set the leader timeout to be 2 standard deviations above the mean. For \FlexRaftIDA, when the leader does
not receive all $N$ responses before the timeout, we repeatedly re-sample until the number of responses meets or exceeds the
previous sample, reflecting the incremental re-encoding mechanism it employs.
We evaluate \sys using a dissemination strategy where the leader initially disseminates a
conservative two fragments per node. Under this dispersal strategy, the leader is able to
tolerate up to $N/4$ missing responses without retransmission. We assume that entries are uncorrelated, and thus an entry at index 501, does not rely on the entry at index 500 to finish before it can complete. 

The graph tracks performance across 1000 invocations of \textsc{Disperse} to the system, with the random latency distribution re-sampled for each entry. Under these circumstances, the baseline is $\sim$1~ms (yellow/orange) and everything above it (darker red) is retransmission overhead due to the leader not receiving enough responses and timing out. Note that due to the randomness employed, each entry's latency can vary slightly.

\begin{figure}
    \centering
    \vspace{-10pt} % 2. Pulls the image up to remove top whitespace
    \includegraphics[width=0.7\linewidth]{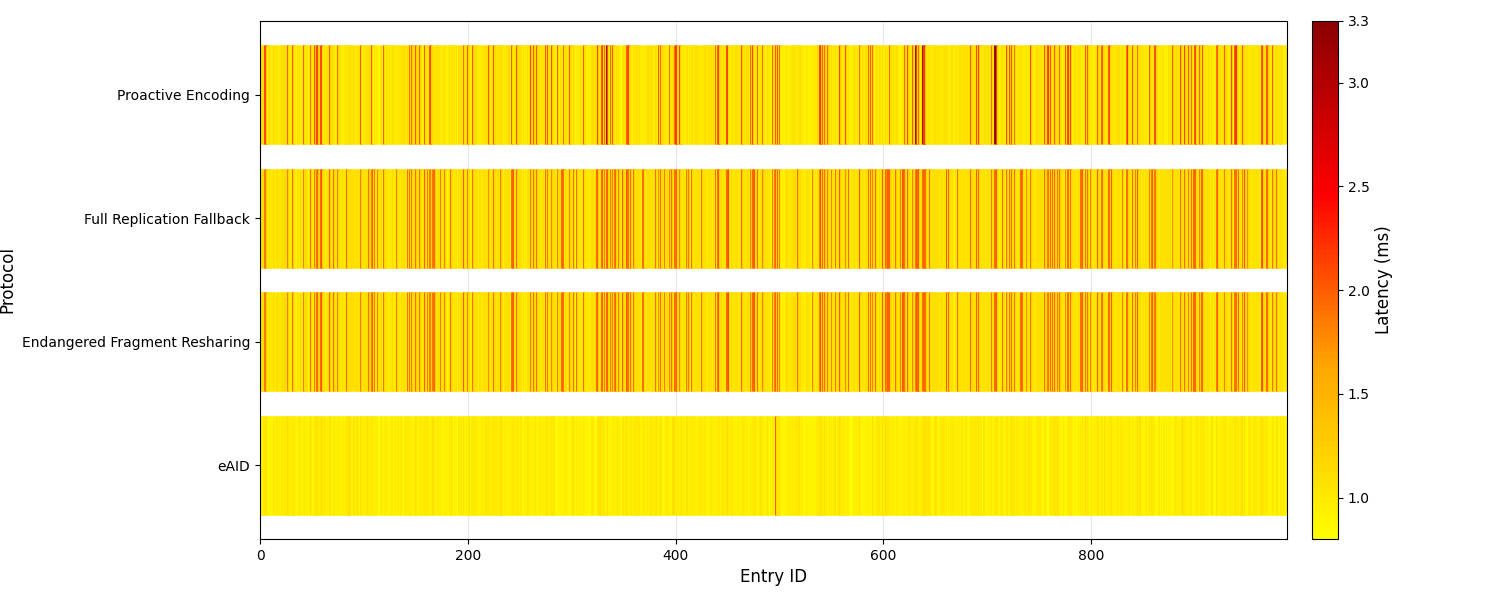} % 3. Set width to \linewidth to fill the wrapper completely
    \vspace{-10pt} % 4. Tighten space between image and caption
    \caption{Analytical evaluation of network latency for Full Replication Fallback, Endangered Fragment Resharing, \FlexRaftIDA, and \sys during unstable network conditions where $N=11$.}
    \label{fig:networkHops}
    \vspace{-15pt} 
\end{figure}

The Full Replication Fallback and Endangered Fragment Resharing 
approaches perform identically, requiring $\sim$2~ms for 22.5\% of 
invocations and $\sim$1~ms for the remainder. \FlexRaftIDA performs 
slightly better, taking $\sim$2~ms for 14.6\% of invocations, but 
its re-encoding approach occasionally triggers a third 
transmission round, pushing latency above 2~ms for 0.4\% of 
invocations. By initially disseminating two fragments per 
node, \sys only requires $\sim$2~ms for less than 1\% of invocations and $\sim$1~ms for the rest. This 
validates Claim \textbf{C3}: \sys outperforms the other approaches 
under unstable network conditions.

\subsection{Key-Value Store over Log Replication (Experimental)}

We deploy both \sysConsensus and \FlexRaft in a 5-node distributed setup and 
evaluate them on YCSB workloads \cite{ycsb} A, B, and C, which span the read/write 
spectrum. Figure~\ref{fig:ycsb-benchmark} reports the average throughput over 5 runs 
as the number of concurrent clients increases from one to four.

Under the write-heavy YCSB-A workload (50\% read / 50\% write), \sysConsensus
consistently outperforms FlexRaft across all client counts, reaching 
$\sim$9.4~Mbps versus $\sim$8.5~Mbps at four clients. This gap stems 
from \sysConsensus's overprovisioned dissemination: it sends extra fragments so it can commit with the fastest 
quorum and prunes the excess post-commit. Since \FlexRaft only sends a single fragment, it must wait on 
every node to safely commit, which stalls reads as well as writes as reads cannot 
return values older than the latest committed write. As the workload shifts to read-dominant 
YCSB-B (95\% read / 5\% write), the gap narrows and both protocols 
converge near 40~Mbps at four clients. Under read-only YCSB-C, the two 
systems are indistinguishable, scaling from $\sim$17~Mbps with one 
client to $\sim$70~Mbps with four. This convergence is expected, as 
there are less writes to bound read progress.

We now evaluate the same setup with only sequential put operations. Since there are no reads to amplify the advantage \sysConsensus offers, the performance is largely the same. 
Figure~\ref{fig:distributed-benchmark} decomposes end-to-end latency
into its two components. The left panel shows end-to-end client
latency, which is dominated by wide-area network round-trip time
and steadily scales as we increase the number of clients. The right panel isolates the
\emph{commit latency} which represents the time spent in the consensus protocol. These results validate \textbf{C4} showing that the overhead of integrating our IDA approach with consensus does not affect performance in the steady state. 
\label{sec:performance}
\begin{figure}
    \centering
    \vspace{-10pt} % 2. Pulls the image up to remove top whitespace
    \includegraphics[width=0.7\linewidth]{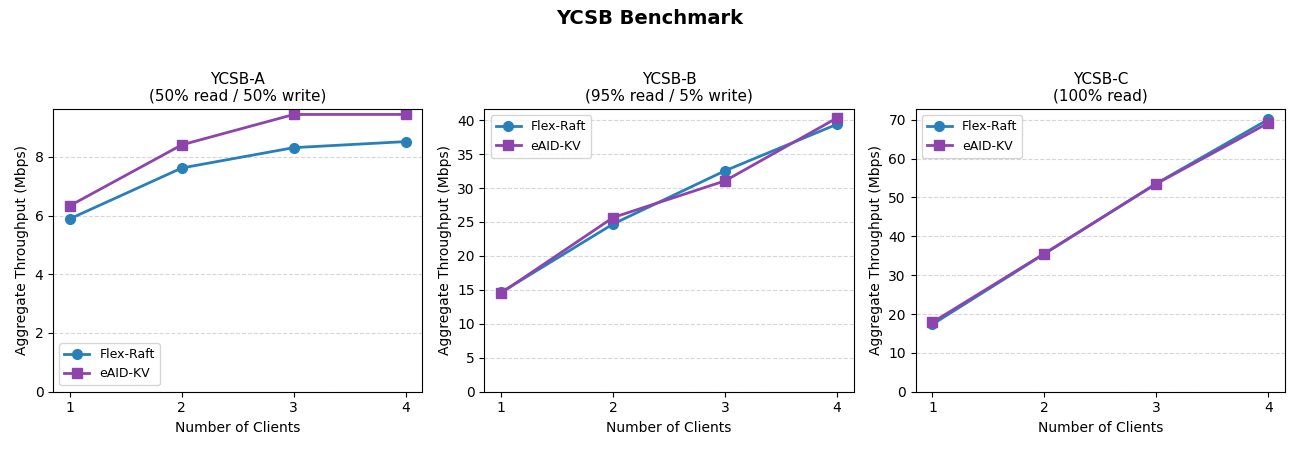} % 3. Set width to \linewidth to fill the wrapper completely
    \vspace{-10pt} % 4. Tighten space between image and caption
    \caption{YCSB benchmark for \FlexRaft vs.\ \sysConsensus}
    \label{fig:ycsb-benchmark}
    \vspace{-5pt} 
\end{figure}

\begin{figure}
    \centering
    \vspace{-10pt} % 2. Pulls the image up to remove top whitespace
    \includegraphics[width=0.7\linewidth]{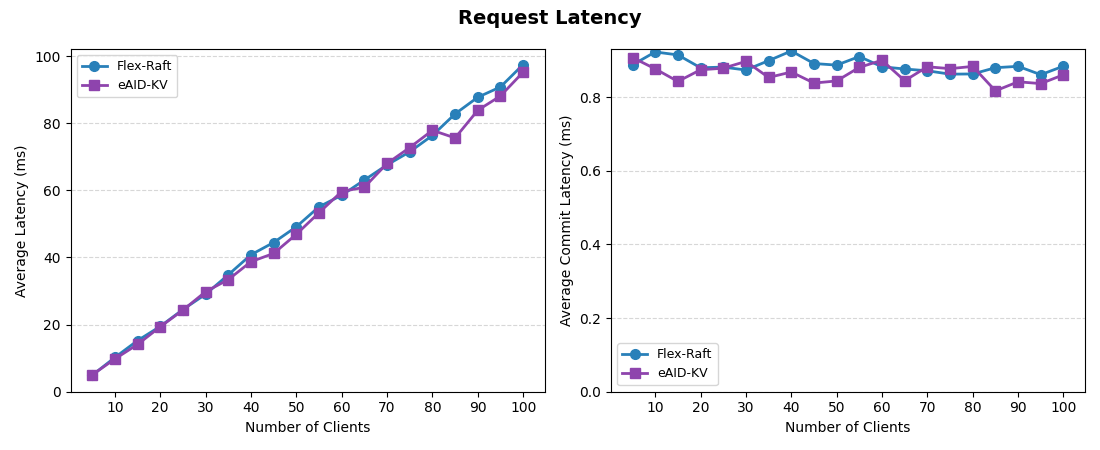} % 3. Set width to \linewidth to fill the wrapper completely
    \vspace{-10pt} % 4. Tighten space between image and caption
    \caption{Distributed benchmark for \FlexRaft vs.\ \sysConsensus}
    \label{fig:distributed-benchmark}
    \vspace{-15pt} 
\end{figure}
\subsection{Summary}
Our evaluation supports all four claims. \sys matches \FlexRaftIDA's storage utilization under normal network conditions (\textbf{C1}) and achieves lower storage post-partition through its autonomous pruning approach (\textbf{C2}). These results carry over to the deployed \sysConsensus system, where the same trends hold experimentally as well. \sys also incurs optimal dispersal latency $\sim$99\% of the time under unstable networks via its adaptive dissemination strategy (\textbf{C3}). Finally, \sysConsensus matches \FlexRaft's end-to-end performance in the steady state (\textbf{C4}), confirming that integrating \sys into a KV store imposes no observable overhead. A useful consequence of this design is that the leader can set aggressive commit timers without stalling on slow nodes: \sysConsensus remains live, with storage utilization converging to the optimal state once slow nodes catch up.

%% file: RelatedWork.tex
\section{Related Work}
\label{sec:RelatedWork}
Replicated \KVs rely on consensus protocols to maintain a consistent state across geographically dispersed nodes. Paxos \cite{Paxos-Simple} established the theoretical foundation for this field, powering systems such as Chubby \cite{Chubby} and Spanner \cite{Spanner}. To address the inherent performance bottlenecks of vanilla Paxos, numerous variants~\cite{WPaxos, EPaxos, PigPaxos, DPaxos, Mencius, RingPaxos} have been proposed. \textbf{Mencius}~\cite{Mencius} utilizes a round-robin approach with rotating leaders to partition the log and balance network load. \textbf{EPaxos}~\cite{EPaxos} adopts a multi-leader design but can only provide a partial ordering (less strict than the total order that other alternatives provide), whereas \textbf{RingPaxos}~\cite{RingPaxos} organizes nodes in a ring topology to exploit high-bandwidth multicast and improve overall throughput. For wide-area deployments, \textbf{WPaxos}~\cite{WPaxos} employs multiple leaders to reduce geographic latency, and \textbf{DPaxos}~\cite{DPaxos} employs dynamic quorums that allow for small replication groups without requiring large, static leader election quorums. Additionally, \textbf{PigPaxos}~\cite{PigPaxos} focuses on the communication layer, using in-network aggregation and piggybacking to decouple the leader's decision-making from message passing.

While the aforementioned works focus on structural optimizations, \textbf{RS-Paxos}~\cite{RS-Paxos} addresses the resource overhead of state machine replication. By integrating erasure coding directly into the Paxos protocol, RS-Paxos replaces traditional full-copy replication with coded fragments. This approach can reduce network and storage costs by over 50\% and significantly improves write throughput by minimizing disk I/O. That said, RS-Paxos introduces a trade-off between fault tolerance and storage efficiency. To tolerate the maximum theoretical number of failures ($F < N/2$), it effectively reverts to full replication, employing erasure coding only when the system is configured to tolerate a strictly lower threshold. Finally, although early storage systems have utilized erasure codes to reduce storage \cite{ErasureCodingAzure}, these systems do not integrate erasure coding with consensus.  

\textbf{Raft}~\cite{Raft} is another widely adopted consensus algorithm designed for understandability, and it is functionally equivalent to Multi-Paxos. Recent research has introduced various optimizations to the original Raft protocol~\cite{ParallelRaft, HoverCraft, RacsSadl}. \textbf{ParallelRaft}~\cite{ParallelRaft} focuses on high concurrency and allows parallel replication of multiple entries as well as out of order execution for non-conflicting commands. To address hardware-level constraints, \textbf{HovercRaft}~\cite{HoverCraft} leverages kernel-bypass techniques and in-network programmability to eliminate CPU and I/O bottlenecks. Finally, \textbf{RACS-SADL}~\cite{RacsSadl} optimizes the performance by separating command dissemination from the core consensus logic. 

\textit{Optimizing with Erasure Coding} Several recent protocols aim to reduce consensus/replication storage and communication cost via the use of erasure coding, including \textbf{\CRaft}~\cite{CRaft}, \textbf{\HRaft}~\cite{HRaft}, \textbf{\FlexRaft}~\cite{FlexRaft}, and \textbf{Crossword}~\cite{hu2025crossword}. We discuss these works in detail in the Background section of the paper (Section~\ref{sec:Background}). Briefly, all of them address only dissemination efficiency and do not support post-dissemination pruning (this has been mentioned as a future challenge~\cite{HRaft}). 
 They administer adjustments to the coding regime only by the leader; and most of them compound the recovery by varying coding regimes across different log entries\footnote{The exception is Crossword, a parallel independent work to \sys employing a $(F+1, (F+1)\times(N-1))$ erasure-code. However, Crossword aims to optimize bandwidth consumption, not storage, and additionally, like all previous works, has no post-dissemination adaptability.}. In contrast, \sys implements post-dissemination pruning, allows each node to locally and autonomously adjust the coding regime per log-entry, and applies a simple and uniform recovery algorithm across all log entries. 
% \Rithwik{Added this sentence, but not sure} 
%\MKR{When you way "this" has been mentioned, what specifically?} \Rithwik{Slightly tangential, but HRaft talks about reducing storage after recovering a failed server}

Additionally, while we focus on crash fault tolerance in our protocol, there are numerous works that target Byzantine failures (e.g., \cite{danezis2025walrusefficientdecentralizedstorage, BeatBFT, ByzantineErasureCodedStorage,HoneyBadgerBFT, DistSimplex, ReduceBandwidthErasureCoding}).

%% file: conclusion.tex
\section{Conclusion}
\label{sec:conclusion}

We presented \sys, an elastic information dispersal algorithm that fixes the recovery
threshold at $F+1$ and varies only the number of fragments per node. \sys decouples dispersal completion from the final storage layout as it completes once
enough responses arrive and prunes surplus fragments in the background, without
coordination. \sysConsensus, the result of integrating \sys into a replicated \KV, matches or exceeds prior systems in the steady state while reclaiming storage after a partition.
Two directions remain open: extending \sys to the Byzantine setting, and exploring
further ways to reduce common-case latency

%% file: Appendix/appendix.tex
\appendix

\input{Appendix/RaftConsensus}

%% file: Appendix/RaftConsensus.tex
\section{Integration with a Replicated Key-Value Store via Log Replication}
\label{sec:kv-integration}

In Section~\ref{sec:use-case} we identified replicated key-value stores 
backed by log replication as the deployment setting for \sys. This 
appendix details that integration. We use Raft~\cite{Raft} as the 
log-replication substrate because it is the standard choice for 
production key-value stores ~\cite{etcd_2021, cockroachdb, tidb} and provides a 
concrete, widely understood target. We first provide an overview of Raft (Appendix \ref{sec:RaftConsensus}), then describe how \sys replaces 
Raft's full-payload \textsc{AppendEntries} with encoded fragment 
dissemination (Appendix \ref{sec:RaftureRaftIntegration}), and finally show 
that this replacement preserves Raft's safety and liveness guarantees 
end-to-end, making \sysConsensus a correct replicated key-value store.

\subsection{Raft Consensus}
\label{sec:RaftConsensus}
\begin{figure}[h!]   
    \centering
    \includegraphics[height=0.6\textheight, keepaspectratio]{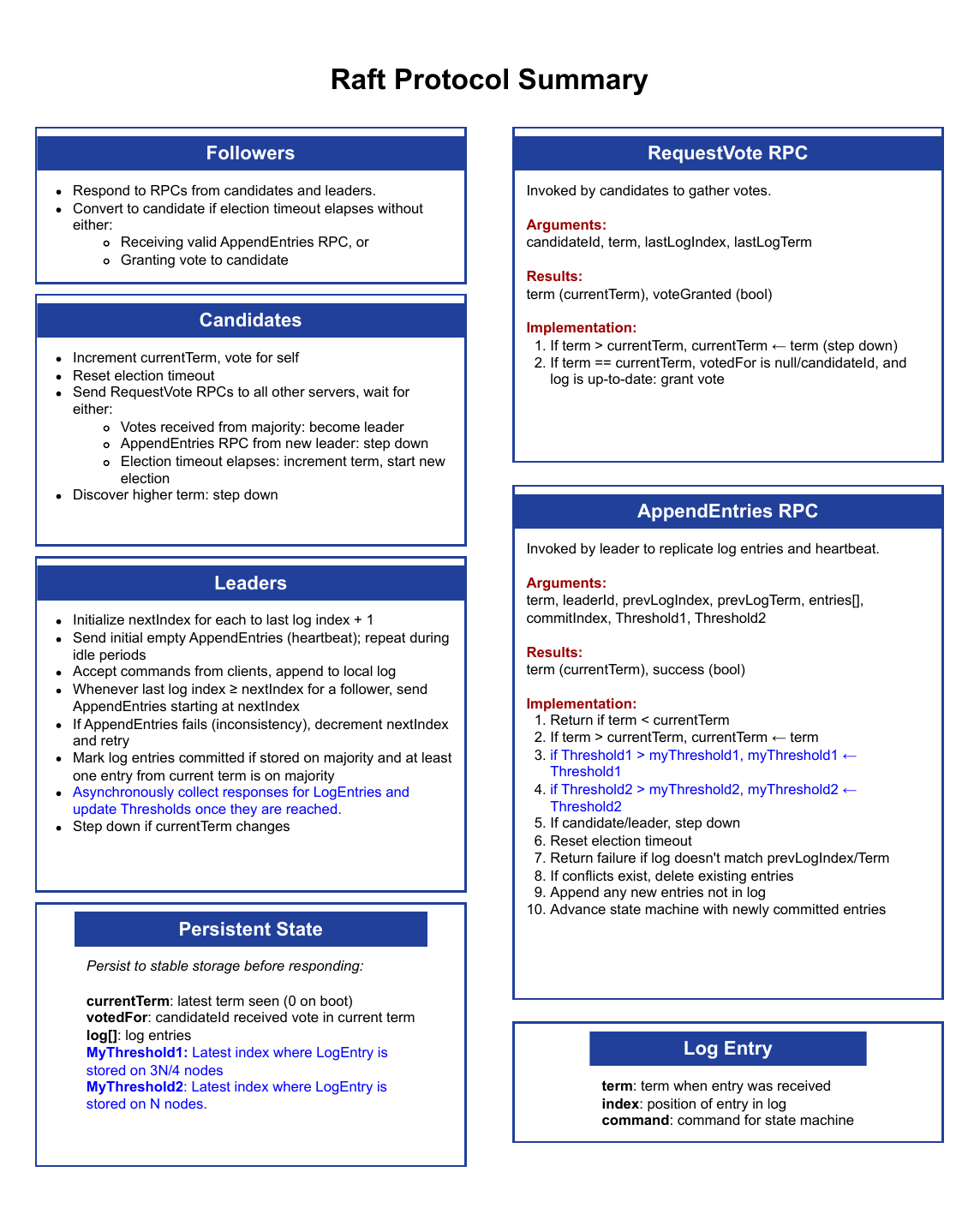}
    \vspace{-1em}
    \caption{Basics of Raft Consensus}
    \label{fig:RaftPseudocode}
\end{figure}

We start by describing the basics of Raft consensus \cite{Raft} which we summarize in Figure \ref{fig:RaftPseudocode}. At a high level, the main goal of Raft consensus is to ensure a fault‑tolerant, consistent replication of a log across a cluster of nodes. A Raft cluster consists of N=2F+1 nodes that can tolerate up to F node failures. At any point of execution, each node in a Raft cluster is categorized as a \textit{leader}, \textit{follower}, or \textit{candidate}. In an ideal point of execution, there is one leader and $N-1$ followers. The leader is in charge of handling all client requests and replicating them to the followers. From each node's perspective, there should be only one viable leader at any point in time. A follower is passive: it only responds to incoming Remote Procedure Calls (RPCs) and does not issue any of its own. Candidates are used to elect a new leader if a majority of nodes believe the current leader has failed. 

Logical time in Raft is divided into \emph{terms} such that there is at most one leader per term and \emph{indices} such that there can be multiple indices per term. Terms are used to identify obsolete information such as log entries from a prior leader that haven't been replicated to a majority of nodes. Indices are used to uniquely identify information within a term such that a (Term, Index) tuple can uniquely identify information in the log. If a follower determines that there is no leader, it changes its state to candidate and sends RequestVote RPCs to all other nodes in the system. The candidate will repeat this until it receives votes from a majority of nodes, in which case it will become the leader. If it receives a response from a valid leader, it will revert to a follower. Lastly, if nobody wins the election, the term is incremented and a new election begins. 

Upon becoming a leader, the node can start responding to client requests. As it receives requests from clients, the leader compiles all of the requests into a LogEntry that also stores meta information such as the term and index. The leader then appends the entry to its log, increments the current index, and sends AppendEntries RPCs to followers. Upon hearing back from a majority, the leader commits the entry, applies the entry to its state machine, and sends a response back to the client. The leader will then notify other nodes of a commit in subsequent RPCs. 

The key component that guarantees Safety in Raft is the log and the rules for electing a leader and adding an entry. For example, a follower will only vote for a candidate if the candidate has a more up-to-date log than it does. Additionally, when the leader sends out AppendEntries, they must include the prior (term, entry) tuple so each follower can verify that its log matches the leader's. Thus if a follower appends a \LogEntry, it will agree on the log with the leader up to and including that entry. 

\subsection{Integration with \sys}
\label{sec:RaftureRaftIntegration}
\input{algorithms/Raft-leader}
\input{algorithms/Raft-follower}

The integration with Raft is structurally similar to the standalone IDA,
with one key difference exploited from Raft's sequential log. Because Raft
guarantees that a follower acknowledging \LogEntry at index $i$ must also
have acknowledged all entries at indices $i' < i$, the leader does not
need to broadcast a per-entry ack count. Instead, it maintains two
monotonically advancing thresholds, $T_1$ and $T_2$: the highest index
acknowledged by at least $\tfrac{3N}{4}$ followers, and the highest index
acknowledged by all $N$ followers, respectively. These thresholds are
piggybacked on subsequent \textsc{AppendEntries} and heartbeat messages,
allowing followers to determine the required fragment count for any prior
entry from $T_1$ and $T_2$ alone (Algorithm~\ref{alg:rafture-follower}).
This compresses the cluster-wide ack-count broadcast of the IDA layer
(Algorithm~\ref{alg:ida-storage}) into two integers, eliminating per-entry
metadata while preserving the pruning guarantees of
Lemma~\ref{lem:pruning-safety}.

We now describe how this integration manifests in three concrete modifications to Raft. 
\paragraph{LogEntry Modifications:} We extend the \texttt{LogEntry} structure to include two new fields: \textbf{Threshold~1} and \textbf{Threshold~2}. Both are tuples in the form $(\mathit{Term}, \mathit{Index})$. \textbf{Threshold~1} represents the highest $(\mathit{Term}, \mathit{Index})$ for which the leader has received responses from $3N/4$ nodes. Similarly, \textbf{Threshold~2} represents the highest $(\mathit{Term}, \mathit{Index})$ for which the leader has received responses from all $N$ nodes.

\paragraph{Log Replication and Commitment:} In \sysConsensus, log replication follows the standard Raft structure but utilizes erasure coding as described in Section~\ref{sec:Solution}.

\paragraph{Follower Processing and Storage Pruning:} Upon receiving an \texttt{AppendEntries} message, the follower node initially persists all associated fragments. To optimize storage usage, it subsequently prunes fragments based on the global resilience indicated by the leader's \textbf{Threshold} values. The pruning strategy maintains a safety invariant: the cluster must always retain sufficient fragments to reconstruct a \LogEntry despite $F < N/2$ failures. As shown in Section \ref{sec:proof}, it is safe for a follower to retain two fragments for all entries indexed until \textbf{Threshold 1}, and a single fragment for all entries indexed until \textbf{Threshold 2}.

\subsection{Safety and Liveness }
We prove Safety and Liveness under the \textbf{Partial Synchrony Model} with $F<N/2$ benign failures. The model functions similarly to the asynchronous model, except there is an unknown Global Stabilization Time (GST), after which all message transmissions have a known bound.
\begin{theorem}[Safety of Raft with \sys]
\label{thm:raft-safety}
Replacing Raft's full-replication \textsc{AppendEntries} payload with
\sys's encoded fragments preserves all five Raft safety invariants.
\end{theorem}

\begin{proof}
\sys modifies only the \emph{payload} of \textsc{AppendEntries} (full
log entries become Reed-Solomon fragments) and the \emph{commit
threshold check}. It does not alter term progression,
voting rules, log indexing, or the leader-election protocol. We address
each invariant.

\emph{Election safety} and \emph{leader append-only} depend solely on
term progression and leader-local log management, neither of which is
modified by \sys. They hold by Raft's original argument.

\emph{Log matching} requires that two logs containing entries with the
same (term, index) agree on all prior entries. \sys preserves the
$(\mathit{term}, \mathit{index}, \mathit{prevLogIndex},
\mathit{prevLogTerm})$ tuple in \textsc{AppendEntries} unchanged
(Algorithm~\ref{alg:rafture-leader}); the consistency check that
followers perform on $\mathit{prevLogIndex}$/$\mathit{prevLogTerm}$
operates on the index/term metadata, not on the fragment payload.
A follower that accepts a fragment for entry $m$ has verified the
prefix at $m-1$ matches the leader's, exactly as in vanilla Raft.

\emph{Leader completeness} requires that any committed entry appears in
the log of every subsequent leader. Raft enforces this via the voting
rule: a candidate must have a log at least as up-to-date as a majority
of voters. In \sys, an entry is committed once the leader collects
$F + t' \ge F + 1$ acks (i.e., a majority). By
Lemma~\ref{lem:dispersal-safety}, at the commit point $M$ is
reconstructable from the cluster's surviving fragments. By
Lemma~\ref{lem:pruning-safety}, this reconstructability is preserved
through arbitrary pruning. Since a future leader is elected from a
majority overlap with the committing majority, and at least one node
in this overlap holds at least one fragment of $M$, the new leader can
collect $F+1$ fragments from the cluster (via \textsc{Reconstruct},
Algorithm~\ref{alg:rafture-follower}) and reconstruct $M$. The leader
then includes $M$ in its log before serving new appends, satisfying
leader completeness.

\emph{State machine safety} follows from leader completeness and
log matching: every node that applies entry $m$ does so against an
identical reconstructed $M$, since reconstruction is deterministic
from any $F+1$ fragments (a property of MDS codes). \qedhere
\end{proof}

\begin{theorem}[Liveness of Raft with \sys]
\label{thm:raft-liveness}
Under partial synchrony with a stable leader after GST, every client
command submitted to the leader is eventually committed.
\end{theorem}

\begin{proof}
Raft's liveness rests on three operational guarantees from its
replication layer: (a) the leader can issue \textsc{AppendEntries}
to followers and receive responses in bounded time after GST;
(b) \textsc{AppendEntries} to a majority of correct followers
eventually succeeds; (c) the leader can determine when an entry is
replicated on a majority. We show \sys preserves all three.

\emph{(a) Bounded response time.} \sys's \textsc{AppendEntries} carries
fragments of size $B/(F+1)$ rather than the full entry of size $B$, and message delivery bounds after GST are preserved.

\emph{(b) Eventual success.} The wait loop inside \textsc{AppendEntry}
(Algorithm~\ref{alg:rafture-leader}) is the Raft-integrated counterpart
to \textsc{Disperse}'s wait loop, and the same termination argument
applies. By Lemma~\ref{lem:ida-termination}, after GST the loop
collects $F + t' \ge F + 1$ acks within bounded time. The retransmission
mechanism converges (Lemma~\ref{lem:ida-termination}) and does not
introduce unbounded waiting.

\emph{(c) Quorum certification.} \sys commits an entry when
$|A[m]| \ge F + t'$, where $F + t' \ge F + 1$ is precisely Raft's
majority threshold. The leader treats \textsc{Dispersed} (i.e., the
return of \textsc{AppendEntry}) as the commit-eligibility signal,
preserving Raft's majority-quorum invariant.

The remainder of Raft's liveness argument --- election timeouts, leader
stability after GST, and progress through the log --- is unchanged by
\sys, since \sys modifies only how a single log entry's payload is
replicated, not the Raft control flow (elections, term progression,
commit-index advancement). Therefore Raft's liveness proof carries
through with \textsc{AppendEntries} replaced by \sys's encoded variant.

Pruning (Lemma~\ref{lem:pruning-safety}) and asynchronous late-ack
collection (\textsc{OnLateAppendEntriesResponse},
Algorithm~\ref{alg:rafture-leader}) occur after a log entry is
committed and do not gate progress on subsequent entries; they affect
only the storage footprint of already-committed entries and therefore
cannot impact liveness. \qedhere
\end{proof}

%% file: algorithms/Raft-leader.tex
\begin{algorithm}[!htbp]
\caption{\sysConsensus: Leader Adaptive Dissemination}
\label{alg:rafture-leader}
\footnotesize
\begin{algorithmic}[1]
\State \textbf{State:} $\mathit{currentTerm}, \mathit{commitIndex}$ (Raft);
       $\mathit{respEst} \gets N$;
       $A[m]$: ack sets;
       $\mathit{sent}[m, j], \mathit{held}[m, j]$: \# frags sent / acked;
       $f_{\text{per}}[m]$: committed per-node threshold for entry $m$;
       $T_1, T_2$: highest indices acked by $\lceil 3N/4 \rceil, N$ followers
\State \textbf{Tiers:} $\mathcal{T} = [(N, 1), (\lceil 3N/4 \rceil, 2), (F{+}1, F{+}1)]$
       \Comment{$\mathrm{Tier}\,1, \mathrm{Tier}\,2, \mathrm{Tier}\,3$ in list order}
\Statex \hspace{1.2em} \Call{Tier}{$q$} returns $(w, f_{\text{per}}) \in \mathcal{T}$ with the largest $w \leq q$
\Statex
\Procedure{AppendEntry}{$M$, index $m$}
    \State $\mathcal{F} \gets$ \Call{Encode}{$M$}, tagged $(\mathit{currentTerm}, m)$;
           $\;(w, f_{\text{per}}) \gets$ \Call{Tier}{$\max(F{+}1, \mathit{respEst} - \delta)$}
    \State Partition $\mathcal{F}$ into $S_1, \ldots, S_N$, $|S_j| = F{+}1$;
           $\;A[m] \gets \{\mathit{leaderId}\}$
    \State Locally store $S_{\mathit{leaderId}}[1{:}f_{\text{per}}]$;
           $\;\mathit{held}[m, \mathit{leaderId}] \gets f_{\text{per}}$
    \For{each follower $j$}
        \State Send \textsc{AppendEntries}$(\mathit{currentTerm}, \mathit{leaderId}, m{-}1, \log[m{-}1].\mathit{term},$
        \Statex \hspace{3em} $\langle (M, S_j[1{:}f_{\text{per}}], \mathit{currentTerm}, m) \rangle, \mathit{commitIndex}, T_1, T_2)$ to $j$
        \State $\mathit{sent}[m, j] \gets f_{\text{per}}$;
               $\;\mathit{held}[m, j] \gets 0$
    \EndFor
    \Repeat
        \State Start timer $\tau$
        \While{$|A[m]| < w$ \textbf{and} timer not expired}
            \State Wait for \textsc{AppendEntriesResponse}$(j, \mathit{term}_j, \mathit{success}, \mathit{held}_j)$
            \If{$\mathit{term}_j > \mathit{currentTerm}$} step down; \textbf{return} \EndIf
            \If{$\mathit{success}$ \textbf{and} $\mathit{held}_j > \mathit{held}[m, j]$}
                $\mathit{held}[m, j] \gets \mathit{held}_j$ \EndIf
            \If{$\mathit{held}[m, j] \geq f_{\text{per}}$}
                $A[m] \gets A[m] \cup \{j\}$; \Call{UpdateThresholds}{$m$} \EndIf
        \EndWhile
        \If{$|A[m]| \geq w$} \textbf{break} \EndIf
        \State $(w, f_{\text{new}}) \gets$ \Call{Tier}{$\max(F{+}1, |A[m]|)$}
               \Comment{timeout: step down a tier}
        \For{each follower $j$ \textbf{with} $f_{\text{new}} > \mathit{sent}[m, j]$}
            \State Send \textsc{AppendEntries}$(\mathit{currentTerm}, \mathit{leaderId}, m{-}1, \log[m{-}1].\mathit{term},$
            \Statex \hspace{3em} $\langle (M, S_j[\mathit{sent}[m, j]{+}1{:}f_{\text{new}}], \mathit{currentTerm}, m) \rangle, \mathit{commitIndex}, T_1, T_2)$ to $j$
            \State $\mathit{sent}[m, j] \gets f_{\text{new}}$
        \EndFor
        \State $f_{\text{per}} \gets f_{\text{new}}$;
               $\;A[m] \gets \{ j \mid \mathit{held}[m, j] \geq f_{\text{per}} \}$
    \Until{false}
    \State $f_{\text{per}}[m] \gets f_{\text{per}}$
           \Comment{freeze per-entry threshold for late-ack accounting}
    \State $\mathit{respEst} \gets |A[m]|$; \Call{Commit}{$m$}
\EndProcedure
\Statex
\Procedure{OnLateAppendEntriesResponse}{$j, m, \mathit{held}_j$}
    \If{$\mathit{held}_j > \mathit{held}[m, j]$}
        \State $\mathit{held}[m, j] \gets \mathit{held}_j$
        \If{$j \notin A[m]$ \textbf{and} $\mathit{held}_j \geq f_{\text{per}}[m]$}
               \Comment{same threshold as the main loop ($f_{\text{per}}$ for entry $m$)}
            \State $A[m] \gets A[m] \cup \{j\}$; \Call{UpdateThresholds}{$m$}
        \EndIf
    \EndIf
\EndProcedure
\Statex
\Procedure{UpdateThresholds}{$m$}
    \If{$|A[m]| \geq \lceil 3N/4 \rceil$ \textbf{and} $m > T_1$} $T_1 \gets m$ \EndIf
    \If{$|A[m]| = N$ \textbf{and} $m > T_2$} $T_2 \gets m$ \EndIf
\EndProcedure
\end{algorithmic}
\end{algorithm}

%% file: algorithms/Raft-follower.tex
\begin{algorithm}[!htbp]
\caption{\sysConsensus: Follower}
\label{alg:rafture-follower}
\footnotesize
\begin{algorithmic}[1]
\State \textbf{State:} $\log[\,]$: local log; each entry stores fragments, term, index
\State $\mathit{currentTerm}, \mathit{commitIndex}$: standard Raft state
\State $\mathit{myT}_1, \mathit{myT}_2$: locally known pruning thresholds
\Statex
\Procedure{OnAppendEntries}{$\mathit{leaderId}$, $\mathit{term}$, $\mathit{prevLogIndex}$, $\mathit{prevLogTerm}$, $\mathit{entries}$, $\mathit{leaderCommit}$, $T_1^{\mathit{ldr}}$, $T_2^{\mathit{ldr}}$}
    \If{$\mathit{term} < \mathit{currentTerm}$}
        \State Send \textsc{AppendEntriesResponse}$(\mathit{self}, \mathit{currentTerm}, \mathit{false}, 0)$ to leaderId;
               \textbf{return} \Comment{(1) stale leader}
    \EndIf
    \If{$\log[\mathit{prevLogIndex}]$ does not exist \textbf{or} $\log[\mathit{prevLogIndex}].\mathit{term} \neq \mathit{prevLogTerm}$}
        \State Send \textsc{AppendEntriesResponse}$(\mathit{self}, \mathit{currentTerm}, \mathit{false}, 0)$ to leaderId;
               \textbf{return} \Comment{(2) log mismatch}
    \EndIf
    \For{each entry $e \in \mathit{entries}$}
        \If{$\log[e.\mathit{index}]$ exists \textbf{and} $\log[e.\mathit{index}].\mathit{term} \neq e.\mathit{term}$}
            \State Truncate $\log$ from $e.\mathit{index}$ onward \Comment{(3) conflicting entry}
        \EndIf
        \If{any fragment in $e$ is not tagged $(e.\mathit{term}, e.\mathit{index})$}
            \State Send \textsc{AppendEntriesResponse}$(\mathit{self}, \mathit{currentTerm}, \mathit{false}, 0)$ to leaderId; \textbf{return}
        \EndIf
        \State Persist $e$'s fragments at $\log[e.\mathit{index}]$ \Comment{(4) append}
    \EndFor
    \If{$\mathit{leaderCommit} > \mathit{commitIndex}$}
        \State $\mathit{commitIndex} \gets \min(\mathit{leaderCommit}, \text{index of last new entry})$ \Comment{(5)}
    \EndIf
    \State $\mathit{myT}_1 \gets \max(\mathit{myT}_1, T_1^{\mathit{ldr}})$;
           $\;\mathit{myT}_2 \gets \max(\mathit{myT}_2, T_2^{\mathit{ldr}})$
    \State \Call{Prune}{} \Comment{opportunistic, runs in background}
    \For{each entry $e \in \mathit{entries}$}
        \State Send \textsc{AppendEntriesResponse}$(\mathit{self}, \mathit{currentTerm}, \mathit{true}, |\log[e.\mathit{index}].\mathit{fragments}|)$ to leaderId
    \EndFor
\EndProcedure
\Statex
\Procedure{Prune}{}
    \For{each entry $m$ in $\log$}
        \If{$m \leq \mathit{myT}_2$}
            \State Retain 1 fragment; discard the rest \Comment{Tier 1: all $N$ nodes hold a fragment}
        \ElsIf{$m \leq \mathit{myT}_1$}
            \State Retain 2 fragments; discard the rest \Comment{Tier 2: $\geq \lceil 3N/4 \rceil$ nodes hold fragments}
        \Else
            \State Retain $F{+}1$ fragments \Comment{Tier 3: bare quorum}
        \EndIf
    \EndFor
\EndProcedure
\Statex
\Procedure{Reconstruct}{entry index $m$}
    \State Collect any $F{+}1$ distinct fragments for $m$ from the cluster, all tagged $(\log[m].\mathit{term}, m)$
    \State Decode $M$ via Reed-Solomon interpolation \Comment{same procedure for every entry}
    \State \textbf{return} $M$
\EndProcedure
\end{algorithmic}
\end{algorithm}